\documentclass[onecolumn, 10pt]{article}
\usepackage[margin=1in]{geometry}

\usepackage[square, comma, sort]{natbib}            
\setcitestyle{numbers}
 \usepackage{graphicx}          

 \usepackage{amssymb}
 \usepackage{amsmath}
\usepackage{thmtools}

\usepackage{amsthm}
\newtheorem{example}{Example}
\newtheorem{remark}{Remark}

\newtheorem{prop}{Proposition}
\newtheorem{thm}{Theorem}
\newtheorem{cor}{Corollary}

\newtheorem{definition}{Definition}
\usepackage{textcomp}
\usepackage{stmaryrd}
\usepackage{mathtools}
 \usepackage{calc}
 \usepackage{tikz}
 \usepackage{pgfplots}
\usepackage{cases}

\usepackage{enumitem}

\usetikzlibrary{arrows}
\usetikzlibrary{shapes}
 \usetikzlibrary{calc}
\usetikzlibrary{decorations.pathreplacing}
\usetikzlibrary{arrows,positioning} 
\usetikzlibrary{pgfplots.groupplots}
\usetikzlibrary{decorations.text}

\newcommand{\Domain}{\mathcal{X}}

\newcommand{\ul}{\underline}
\newcommand{\ol}{\overline}
\newcommand{\Kinf}{\mathcal{K}_\infty}
\newcommand*\wc{{\mkern 2mu\cdot\mkern 2mu}}

\usetikzlibrary{decorations.markings}
\newcommand{\X}{\mathcal{X}}
\renewcommand\footnotemark{}

\begin{document}
\title{A Contractive Approach to Separable Lyapunov Functions for Monotone Systems}
\thanks{School of Electrical and Computer Engineering and School of Civil and Environmental Engineering\\Georgia Institute of Technology. \texttt{sam.coogan@gatech.edu}. }
\author{Samuel Coogan} 
\date{}
\maketitle

\begin{abstract}
Monotone systems preserve a partial ordering of states along system trajectories and are often amenable to separable Lyapunov functions that are either the sum or the maximum of a collection of functions of a scalar argument. In this paper, we consider constructing separable Lyapunov functions for monotone systems that are also contractive, that is, the distance between any pair of trajectories exponentially decreases. The distance is defined in terms of a possibly state-dependent norm. When this norm is a weighted one-norm, we obtain conditions which lead to sum-separable Lyapunov functions, and when this norm is a weighted infinity-norm, symmetric conditions lead to max-separable Lyapunov functions. In addition, we consider two classes of Lyapunov functions: the first class is separable along the system's state, and the second class is separable along components of the system's vector field. The latter case is advantageous for many practically motivated systems for which it is difficult to measure the system's state but easier to measure the system's velocity or rate of change. In addition, we present an algorithm based on sum-of-squares programming to compute such separable Lyapunov functions. We provide several examples to demonstrate our results.

\end{abstract}

\section{Introduction}

If a dynamical system maintains a partial ordering of states along trajectories of the system, it is said to be \emph{monotone} \cite{Hirsch:1983lq,Hirsch:1985fk, Smith:2008fk}. Large classes of physically motivated systems have been shown to be monotone including biological networks \cite{Sontag:2007ad} and transportation networks \cite{Gomes:2008fk, Lovisari:2014yq, coogan2015compartmental}.  Monotone systems exhibit structure and ordered behavior that is exploited for analysis and control \emph{e.g.}, \cite{Angeli:2003fv, Angeli:2004qy, Rantzer:2012fj, Dirr:2015rt}.  

%
%

It has been observed that monotone systems are often amenable to separable Lyapunov functions for stability analysis. In particular, classes of monotone systems have been identified that allow for Lyapunov functions that are the sum or maximum of a collection of functions of a scalar argument \cite{Rantzer:2013bf, Dirr:2015rt, Sootla:2016sp}. In the case of linear monotone systems, also called \emph{positive} systems, such sum-separable and max-separable global Lyapunov functions are always possible when the origin is an asymptotically stable equilibrium \cite{Rantzer:2012fj}. 

On the other hand, if the distance between states along any pair of trajectories is exponentially decreasing, a dynamical system is said to be \emph{contractive} \cite{Pavlov:2004lr, LOHMILLER:1998bf, Sontag:2010fk,Forni:2012qe}. If a contractive system has an equilibrium, then the equilibrium is globally asymptotically stable and a Lyapunov function is given by the distance to the equilibrium. 

Certain classes of monotone systems have been shown to be also contractive with respect to non-Euclidean norms. For example, \cite{Margaliot:2012hc, Margaliot:2014qv, Raveh:2015wm} study a model for gene translation which is monotone and contractive with respect to a weighted $\ell_1$ norm. A closely related result is obtained for transportation flow networks in \cite{Coogan:2014ph, Como:2015ne}. In \cite{Coogan:2014ph}, a Lyapunov function defined as the magnitude of the vector field is used, and in \cite{Como:2015ne}, a Lyapunov function based on the distance of the state to the equilibrium is used.

In this paper, we establish sufficient conditions for constructing  sum-separable and max-separable Lyapunov functions for monotone systems by appealing to contraction theory. In particular, we study monotone systems that are contractive with respect to a possibly state-dependent, weighted $\ell_1$ norm, which leads to sum-separable Lyapunov functions, or weighted $\ell_\infty$ norm, which leads to max-separable Lyapunov functions. We first provide sufficient conditions establishing contraction for monotone systems in terms of negativity of scaled row or column sums of the Jacobian matrix for the system where the scaling may be state-dependent.

In addition to deriving Lyapunov functions that are separable along the state of the system, we also introduce Lyapunov functions that are separable along components of the vector field. This is especially relevant for certain classes of systems such as multiagent control systems or flow networks where it is often more practical to measure velocity or flow rather than position or state.   Additionally, we present results of independent interest for proving asymptotic stability and obtaining Lyapunov functions of systems that are \emph{nonexpansive} with respect to a particular vector norm, \emph{i.e.}, the distance between states along any pair of trajectories does not increase. Finally, we draw connections between our results and related results, particularly small-gain theorems for interconnected input-to-state stable (ISS) systems.

The present paper significantly generalizes results previously reported in \cite{Coogan:2016kx}, which only considered constant norms over the state-space. Thus, the theorems of \cite{Coogan:2016kx} are presented here as corollaries to our main results. In addition to deriving separable Lyapunov functions for a class of monotone systems, we present explicit conditions for establishing contractive properties for nonautonomous, nonlinear systems with respect to state-dependent, non-Euclidean metrics. These conditions require that the matrix measure of a suitably defined generalized Jacobian matrix be uniformly negative definite and rely on the theory of Finsler-Lyapunov functions \cite{Forni:2012qe}.  

The recent work \cite{Manchester:2017wc} also seeks to determine when a contracting monotone system has a separable contraction metric. Unlike the present work, which focuses on norms that are naturally separable, namely, the $\ell_1$ and $\ell_\infty$ norm, \cite{Manchester:2017wc} considers Riemannian metrics and studies when a contracting system also has a separable Riemannian contraction metric.


This paper is organized as follows. Section \ref{sec:notation} defines notation and Section \ref{sec:problem-setup} provides the problem setup.  Section \ref{sec:main-results} contains the statements of our main results. Before proving these results, we review contraction theory for general, potentially time-varying nonlinear systems in Section \ref{sec:contr-with-resp}. 
In applications, it is often the case that the system dynamics are not contractive everywhere, but are nonexpansive with respect to a particular norm, \emph{i.e.}, the distance between any pair of trajectories does not increase for all time. In Section \ref{sec:glob-asympt-stab}, we provide a sufficient condition for establishing global asymptotic stability for nonexpansive systems. 

In Section \ref{sec:proof-main-result}, we provide the proofs of our main results, and in Section \ref{sec:an-algor-comp}, we use our main results to establish a numerical algorithm for searching for separable Lyapunov functions using sum-of-squares programming. We provide several applications of our results in Section \ref{sec:examples}. Next, we discuss how the present work relates to results for interconnected ISS systems and generalized contraction theory in Section \ref{sec:disc-comp-exist}. Section \ref{sec:conclusions} contains concluding remarks.

\section{Preliminaries}
\label{sec:notation}

For $x,y\in\mathbb{R}^n$, we write $x\leq y$ if the inequality holds elementwise, and likewise for $<$, $\geq$ and $>$. Similarly, $x>0$  ($x\geq 0$)  means all elements of $x$ are positive (nonegative), and likewise for $x<0$ or $x\leq 0$. For symmetric $X,Y\in\mathbb{R}^{n\times n}$, $X\succ 0$ (respectively, $X\succeq 0$) means $X$ is positive definite (respectively, semidefinite), and $X\succ Y$ (respectively, $X\succeq Y$) means $X-Y\succ 0$ (respectively, $X-Y\succeq 0$). For $\X\subseteq \mathbb{R}^n$, the matrix-valued function $\Theta:\X\to\mathbb{R}^{n\times n}$ is \emph{uniformly positive definite on $\X$} if $\Theta(x)=\Theta(x)^T$ for all $x\in\X$ and there exists $\alpha >0$ such that $\Theta(x)\succeq \alpha I$ for all $x\in\X$ where $I$ denotes the identity matrix and $T$ denotes transpose.
%

The vector of all ones is denoted by $\mathbf{1}$. For functions of one variable, we denote derivative with the prime notation $'$. The $\ell_1$ and $\ell_\infty$ norms are denote by $|\cdot|_1$ and $|\cdot|_\infty$, respectively, that is, $|x|_1=\sum_{i=1}^n|x_i|$ and $|x|_\infty=\max_{i=1,\ldots,n}|x_i|$ for $x\in\mathbb{R}^n$.

Let $|\wc|$ be some vector norm on $\mathbb{R}^n$ and let $\|\wc \|$ be its induced matrix norm on $\mathbb{R}^{n\times n}$. The corresponding \emph{matrix measure} of the matrix $A\in\mathbb{R}^{n\times n}$ is defined as (see, \emph{e.g.}, \cite{Desoer:2008bh})
\begin{align}
  \label{eq:5}
  \mu(A):= \lim_{h\to 0^+}\frac{\|I+hA\|-1}{h}.
\end{align}

One useful property of the matrix measure is that $\mu(A)<0$ implies $A$ is Hurwitz \cite{Vidyasagar:2002ly}. 
For the $\ell_1$ norm, %
the induced matrix measure is given by
\begin{align}
  \label{eq:56}
  \textstyle \mu_1(A)=\max_{j=1,\ldots,n}\left(A_{jj}+\sum_{i\neq j}|A_{ij}|\right)
\end{align}
for any $A\in\mathbb{R}^{n\times n}$. Likewise, for the $\ell_\infty$ norm, %
the induced matrix measure is given by
\begin{align}
  \label{eq:58}
  \textstyle   \mu_\infty(A)=\max_{i=1,\ldots,n}\left(A_{ii}+\sum_{j\neq i}|A_{ij}|\right).
\end{align}
See, \emph{e.g.}, \cite[Section II.8, Theorem 24]{Desoer:2008bh}, for a derivation of the induced matrix measures for common vector norms.

A matrix $A\in\mathbb{R}^{n\times n}$ is \emph{Metzler} if all of its off diagonal components are nonnegative, that is, $A_{ij}\geq 0$ for all $i\neq j$.  When $A$ is Metzler, %
\eqref{eq:56} and \eqref{eq:58} reduce to%
\begin{align}
  \label{eq:10}
  \mu_1(A)&\textstyle =\max_{j=1,\ldots,n}\sum_{i=1}^nA_{ij},\\
  \label{eq:10-2}  \mu_\infty(A)&\textstyle =\max_{i=1,\ldots,n}\sum_{j=1}^nA_{ij},%
\end{align}
that is, $\mu_1(A)$ is the largest column sum of $A$ and $\mu_{\infty}(A)$ is the largest row sum of $A$.

\section{Problem Setup}
\label{sec:problem-setup}

Consider the dynamical system
\begin{align}
  \label{eq:1}
  \dot{x}=f(x)
\end{align}
for $x\in\Domain\subseteq \mathbb{R}^n$ where %
$f(\cdot)$ is continuously differentiable. Let $f_i(x)$ indicate the $i$th component of $f$ and let  $J(x)= \frac{\partial f}{\partial x}(x)$ be the Jacobian matrix of $f$.

Throughout, when $\mathcal{X}$ is not an open set (\emph{e.g.}, $\X$ is the positive orthant as in Example \ref{ex:statedep} below), we understand continuous differentiability of a function on $\mathcal{X}$ to mean that  the function can be extended to a continuously differentiable function on some open set containing $\X$, and we further consider $\X$ to be equipped with the induced topology.


%
%
 
%
%

Denote by $\phi(t,x_0)$ the solution to \eqref{eq:1} at time $t$ when the system is initialized with state $x_0$ at time $0$. We assume that \eqref{eq:1} is forward complete and $\Domain$ is forward invariant for \eqref{eq:1} so that $\phi(t,x_0)\in \Domain$ for all $t\geq 0$ and all $x_0\in\Domain$.  %

Except in Sections \ref{sec:contr-with-resp} and \ref{sec:glob-asympt-stab}, we assume \eqref{eq:1} is monotone \cite{Smith:2008fk, Angeli:2003fv}:
\begin{definition}
\label{def:mon}
The system \eqref{eq:1} is \emph{monotone} if the dynamics maintain a partial order on solutions, that is,
\begin{align}
  \label{eq:3}
  x_0\leq y_0 \implies \phi(t,x_0)\leq \phi(t,y_0) \quad \forall t\geq 0
\end{align}  
for any $x_0,y_0\in\Domain$.
\end{definition}
In this paper, monotonicity is defined with respect to the positive orthant since inequalities are interpreted componentwise, although it is common to consider monotonicity with respect to other cones \cite{Angeli:2003fv}.

The following proposition characterizes monotonicity in terms of the Jacobian matrix $J(x)$. 

\begin{prop}[{\emph{Kamke Condition}, \cite[Ch. 3.1]{Smith:2008fk}}]
\label{prop:mono}
Assume $\Domain$ is convex. Then the system \eqref{eq:1} is monotone if and only if the Jacobian $J(x)$ is Metzler for all $x\in\Domain$.  
\end{prop}

Here, we are interested in certifying stability of an equilibrium $x^*$ for the dynamics \eqref{eq:1}. To that end, we have the following definition of Lyapunov function, the existence of which implies asymptotic stability of $x^*$.

\begin{definition}[{\cite[p. 219]{Sontag:1998cr}}]
\label{def:lyap}
Let $x^*$ be an equilibrium of \eqref{eq:1}.  A continuous function $V:\Domain\to\mathbb{R}$ is a \emph{(local) Lyapunov function} for \eqref{eq:1} with respect to $x^*$ if on some neighborhood $\mathcal{O}$ of $x^*$ the following hold:
  \begin{enumerate}[label=(\roman*)]
  \item $V(x)$ is proper at $x^*$, that is, for small enough $\epsilon>0$, $\{x\in\Domain|V(x)\leq \epsilon\}$ is a compact subset of $\mathcal{O}$;
\item $V(x)$ is positive definite on $\mathcal{O}$, that is, $V(x)\geq 0$ for all $x\in\mathcal{O}$ and $V(x)=0$ if and only if $x=x^*$;
\item For any $x_0\in \mathcal{O}$, $x_0\neq x^*$, there is some time $\tau>0$ such that $V(\phi(\tau,x_0))<V(x_0)$ and $V(\phi(t,x_0))\leq V(x_0)$ for all $t\in(0,\tau]$.
  \end{enumerate}
Furthermore, $V$ is a global Lyapunov function for \eqref{eq:1} if we may take $\mathcal{O}=\mathcal{X}$ and $V(x)$ is globally proper, that is, for each $L>0$, $\{x\in\Domain|V(x)\leq L\}$ is compact.
\end{definition}

%

In this paper, we are particularly interested in Lyapunov functions defined using nondifferentiable norms, thus we rely on the definition above which only requires $V(x)$ to be continuous. Such nondifferentiable Lyapunov functions will not pose additional technical challenges since we will not rely on direct computation of the gradient of $V(x)$.  Instead, we will construct locally Lipschitz continuous Lyapunov functions and bound the time derivative evaluated along trajectories of the system, which exists for almost all time. Nonetheless, classical Lyapunov theory, which verifies condition (3) of the above definition by requiring $(\partial V/\partial x)\cdot f(x)<0$ for all $x\neq x^*$, is extended to such nondifferentiable Lyapunov functions with the use of generalized derivatives \cite{Clarke:1990kx}. %

Note that when $\Domain=\mathbb{R}^n$, a Lyapunov function is globally proper (and hence a global Lyapunov function) if and only if it is radially unbounded \cite[p. 220]{Sontag:1998cr}.

We call the Lyapunov function $V(x)$ \emph{agent sum-separable} if it decomposes as
\begin{align}
  \label{eq:61}
  V(x)=\sum_{i=1}^nV_i(x_i,f_i(x))
\end{align}
for a collection of functions $V_i$, and \emph{agent max-separable} if it decomposes as
\begin{align}
  \label{eq:62}
  V(x)=\max_{i=1,\ldots,n}V_i(x_i,f_i(x)).
\end{align}
If each $V_i$ in \eqref{eq:61} (respectively, \eqref{eq:62}) is a function only of $x_i$, we further call $V(x)$ \emph{state sum-separable} (respectively, \emph{state max-separable}). On the other hand, if each $V_i$ is a function only of $f_i(x)$, we say $V(x)$ is \emph{flow state-separable} (respectively, \emph{flow max-separable}).

Our objective is to construct separable Lyapunov functions for a class of monotone systems.

\section{Main Results}
\label{sec:main-results}
In the main result of this paper, we provide conditions for certifying that a monotone system possesses a globally asymptotically stable equilibrium. These conditions then lead to easily constructed separable Lyapunov functions. As discussed in Section \ref{sec:contr-with-resp}, this result relies on establishing that the monotone system is contractive with respect to a possibly state-dependent norm. We will first state our main result in this section before developing the theory required for its proof in the sequel.%

\begin{definition}
  The set $\mathcal{X}\subseteq \mathbb{R}^n$ is \emph{rectangular} if %
 there exists a collection of connected sets (\emph{i.e.}, intervals) $\Domain_i\subseteq \mathbb{R}$ for $i=1,\ldots, n$ such that $\Domain=\Domain_1\times \cdots \times \Domain_n$.

\end{definition}

\begin{thm}
\label{thm:tvmain1}
Let \eqref{eq:1} be a monotone system with rectangular domain $\mathcal{X}=\prod_{i=1}^n\X_i$, and let $x^*\in\X$ be an equilibrium for \eqref{eq:1}. If there exists a collection of continuously differentiable functions $\theta_i:\Domain_i\to\mathbb{R}$ for $i=1,\ldots,n$ such that for some $c>0$, $\theta_i(x_i)\geq c$ for all $x_i\in \X_i$ and for all $i$, and 
\begin{align}
  \label{eq:111}
    \theta(x)^TJ(x)+\dot{\theta}(x)^T&\leq 0  \quad \forall x\in\mathcal{X},\\
\label{eq:114} \theta(x^*)^TJ(x^*)&< 0,
\end{align}
where $\theta(x)=\begin{bmatrix}\theta_1(x_1) &\theta_2(x_2)&\ldots&\theta_n(x_n)\end{bmatrix}^T$, 
 %
%
then $x^*$ is globally asymptotically stable. %
Furthermore,
\begin{align}
  \label{eq:113}
\sum_{i=1}^n\left|\int_{x_i^*}^{x_i}\theta_i(\sigma)d\sigma\right|
\end{align}
is a global Lyapunov function, and
\begin{align}
  \label{eq:113-2} \sum_{i=1}^n\theta_i(x_i)|f_i(x)|
\end{align}
is a local Lyapunov function. If \eqref{eq:113-2} is globally proper, then it is also a global Lyapunov function.

\end{thm}
Above, $\dot{\theta}(x)$ is shorthand for 
\begin{align}
  \label{eq:163}
\dot{\theta}(x)=\frac{\partial \theta}{\partial x}f(x)=%
\begin{bmatrix}
  \theta_1'(x_1)f_1(x)&\ldots\    \theta_n'(x_n)f_n(x)
\end{bmatrix}^T.  
\end{align}

\begin{thm}
  \label{thm:tvmain2}
Let \eqref{eq:1} be a monotone system with rectangular domain $\mathcal{X}=\prod_{i=1}^n\X_i$, and let $x^*\in\X$ be an equilibrium for \eqref{eq:1}. If there exists a collection of continuously differentiable functions $\omega_i:\Domain_i\to\mathbb{R}$ for $i=1,\ldots,n$ such that for some $c>0$, $0<\omega_i(x_i)\leq c$ for all $x_i\in\X_i$ and for all $i$, and
\begin{align}
  \label{eq:45}  J(x)\omega(x)-\dot{\omega}(x)&\leq 0 \quad \forall x\in\mathcal{X},\\
  \label{eq:48} J(x^*)\omega(x^*)&< 0,
\end{align}
where $\omega(x)=\begin{bmatrix}\omega_1(x_1) &\omega_2(x_2)&\ldots&\omega_n(x_n)\end{bmatrix}^T$, 


then $x^*$ is globally asymptotically stable.
Furthermore, %
\begin{align}
  \label{eq:47}
    \max_{i=1,\ldots,n}\left|\int_{x_i^*}^{x_i}\frac{1}{\omega_i(\sigma)}d\sigma\right|
\end{align}
is a global Lyapunov function and 
\begin{align}
\label{eq:47-2} \max_{i=1,\ldots,n}\frac{1}{\omega_i(x_i)}|f_i(x)|
\end{align}
is a local Lyapunov function. If \eqref{eq:47-2} is globally proper, then it is also a global Lyapunov function.
\end{thm}
Note that \eqref{eq:113} and \eqref{eq:47} are state-separable Lyapunov functions, while \eqref{eq:113-2} and \eqref{eq:47-2} are agent-separable Lyapunov functions.

\begin{example}
\label{ex:statedep}
  Consider the system
  \begin{align}
    \label{eq:80}
    \dot{x}_1&=-x_1+x_2^2\\
    \dot{x}_2&=-x_2
  \end{align}
which is monotone on the invariant domain $\Domain=(\mathbb{R}_{\geq 0})^2$ with unique equilibrium $x^*=(0,0)$. The Jacobian is given by
\begin{align}
  \label{eq:86}
  J(x)=
  \begin{bmatrix}
    -1&2x_2\\
    0&-1
  \end{bmatrix}.
\end{align}
Let $\theta(x)=
\begin{bmatrix}
  1&x_2+1
\end{bmatrix}^T$ so that $\dot{\theta}(x)=
\begin{bmatrix}
  0&-x_2
\end{bmatrix}^T$ and
\begin{align}
  \label{eq:90}
  \theta(x)^TJ(x)+\dot{\theta}(x)^T&=
  \begin{bmatrix}
    -1&-1
  \end{bmatrix}\leq 0
\end{align}
for all $x\in\Domain$. Since also $\theta(x^*)^TJ(x^*)=  \begin{bmatrix}
    -1&-1
  \end{bmatrix}$, Theorem~\ref{thm:tvmain1} is applicable. Therefore, $x^*=0$ is globally asymptotically stable and, from \eqref{eq:113} and \eqref{eq:113-2},  either of the following are global Lyapunov functions:
  \begin{align}
    \label{eq:93}
    V_1(x)%
&=x_1+x_2+\frac{1}{2}x_2^2\text{, and}\\
V_2(x)%
&=\left|{-x_1}+x_2^2\right|+x_2+x_2^2.
  \end{align}
On the other hand, let $\omega(x)=
\begin{bmatrix}
  2&\frac{1}{1+x_2}
\end{bmatrix}^T$ for which $\dot{\omega}(x)=\begin{bmatrix}0&\frac{x_2}{(1+x_2)^2}\end{bmatrix}^T$ and
\begin{align}
  \label{eq:94}
  J(x)\omega(x)-\dot{\omega}(x)=
  \begin{bmatrix}
    -2+\frac{2x_2}{1+x_2}\\
\frac{-1}{1+x_2}-\frac{x_2}{(1+x_2)^2}
  \end{bmatrix}\leq 0
\end{align}
for all $x\in\Domain$. Since also $J(x^*)\omega(x^*)=\begin{bmatrix}-2&-1\end{bmatrix}^T$, Theorem~\ref{thm:tvmain2} is applicable. 
Thus, from \eqref{eq:47} and \eqref{eq:47-2}, either of the following are also Lyapunov functions:
\begin{align}
  \label{eq:95}
  V_3(x)%
&=\max\left\{\frac{1}{2}x_1,x_2+\frac{1}{2}x_2^2\right\}{\text{, and}}\\
V_4(x)%
&=\max\left\{\frac{1}{2}\left|{-x_1}+x_2^2\right|,x_2+x_2^2\right\}.
\end{align}
$\blacksquare$
\end{example}

%


\begin{remark}
  Even when \eqref{eq:113-2} or \eqref{eq:47-2} is not globally proper and thus not a global Lyapunov function, it is nonetheless the case that both functions monotonically decrease to zero along any trajectory of the system.
\end{remark}

We now specialize Theorems \ref{thm:tvmain1} and \ref{thm:tvmain2} to the case where $\theta(x)$ and $\omega(x)$ are independent of $x$, that is, are constant vectors. This special case proves to be especially useful in a number of applications as demonstrated in Section \ref{sec:examples}.

\begin{cor}
\label{cor:1}
    Let \eqref{eq:1} be a monotone system with equilibrium $x^*$. Suppose there exists a vector $v>0$ such that $v^TJ(x)\leq 0$ for all $x\in\mathcal{X}$ and {$v^TJ(x^*)<0$. Then $x^*$ is globally asymptotically stable}. 
Furthermore, %
\begin{align}
  \label{eq:17}
  \sum_{i=1}^nv_i|x_i-x_i^*|%
\end{align}
is a global Lyapunov function and
\begin{align}
\label{eq:17-2}
\sum_{i=1}^nv_i|f_i(x)|
\end{align}
is a local Lyapunov function. {If, additionally, there exists $c>0$ such that $v^TJ(x)\leq -c\mathbf{1}^T$ for all $x\in\mathcal{X}$, then \eqref{eq:17-2} is also a global Lyapunov function.}

\end{cor}

\begin{cor}
\label{cor:2}
      Let \eqref{eq:1} be a monotone system with equilibrium $x^*$. Suppose there exists a vector $w>0$ such that $J(x)w\leq 0$ for all $x\in\mathcal{X}$ and {$J(x^*)w<0$. Then $x^*$ is globally asymptotically stable}. 
Furthermore, %
\begin{align}
  \label{eq:19}
  \max_{i=1,\ldots,n}\left\{\frac{1}{w_i}|x_i-x_i^*|\right\}
\end{align}
is a global Lyapunov function and
\begin{align}
\label{eq:20}
  \max_{i=1,\ldots,n}\left\{\frac{1}{w_i}|f_i(x)|\right\}
\end{align}
is a local Lyapunov function. {If, additionally, there exists $c>0$ such that $J(x)w\leq -c\mathbf{1}$ for all $x\in\mathcal{X}$, then \eqref{eq:20} is also a global Lyapunov function.}

\end{cor}


%

%

Note that \eqref{eq:17} and \eqref{eq:19} are state separable Lyapunov functions while \eqref{eq:17-2} and \eqref{eq:20} are flow separable Lyapunov functions. {\text{Corollaries \ref{cor:1} and \ref{cor:2}} were previously reported in \cite{Coogan:2016kx}}.

The following example shows that Corollaries \ref{cor:1} and \ref{cor:2} recover a well-known condition for stability of monotone linear systems.
 \begin{example}[Linear systems]
\label{ex:1}
Consider $\dot{x}=Ax$ for $A$ Metzler. Corollaries \ref{cor:1} and \ref{cor:2} imply that if one of the following conditions holds, then the origin is globally asymptotically stable:
\begin{align}
  \label{eq:53}
  &\text{There exists $v>0$ such that $v^TA<0$, \quad or}\\
  \label{eq:53-2}&\text{There exists $w>0$ such that $Aw<0$}.
\end{align}
If \eqref{eq:53} holds then $\sum_{i=1}^nv_i|x_i|$ and $\sum_{i=1}^n v_i|(Ax)_i|$ are Lyapunov functions, and if \eqref{eq:53-2} holds then $\max_i\{|x_i|/w_i\}$ and $\max_i\{|(Ax)_i|/w_i\}$ are Lyapunov functions where $(Ax)_i$ denotes the $i$th element of $Ax$. \hfill $\blacksquare$
\end{example}
In fact, it is well known that $A$ is Hurwitz if and only if either (and, therefore, both) of the two conditions \eqref{eq:53} and \eqref{eq:53-2} hold, as noted in, \emph{e.g.}, \cite [Proposition 1]{Rantzer:2012fj}, and the corresponding state separable Lyapunov functions of Example \ref{ex:1} are also derived in \cite{Rantzer:2012fj}. Thus, Theorems \ref{thm:tvmain1} and \ref{thm:tvmain2}, along with Corollaries \ref{cor:1} and \ref{cor:2}, are considered as nonlinear extensions of these results.

The proofs of Theorems \ref{thm:tvmain1} and \ref{thm:tvmain2} use contraction theoretic arguments and, specifically, show that a monotone system satisfying the hypotheses of the theorems is contractive with respect to a suitably defined, state-dependent norm. The proof technique illuminates useful properties of contractive systems that are of independent interest and appear to be novel. These results are presented next, before returning to the proofs of the above theorems.%

\section{Contraction with respective to state-dependent, non-Euclidean metrics}
\label{sec:contr-with-resp}
In the following two sections, we develop preliminary results necessary to prove our main results of Section \ref{sec:main-results}.   Here, we do not require the monotonicity property. Moreover, we develop our results for potentially time-varying systems, that is, systems of the form $\dot{x}=f(t,x)$.%

We first provide conditions for establishing that a nonlinear system is \emph{contractive}. A system is contractive with respect to a given metric if the distance between any two trajectories decreases at an exponential rate. Contraction with respect to Euclidean norms with potentially state-dependent weighting matrices is considered in \cite{LOHMILLER:1998bf}. This approach equips the state-space with a Riemannian structure. For metrics defined using non-Euclidean norms, which have proven useful in many applications, a Riemannian approach is insufficient, and contraction has been characterized using matrix measures for fixed (i.e., state-independent) norms \cite{Sontag:2010fk}. These approaches were recently unified and generalized in \cite{Forni:2012qe} using the theory of Finsler-Lyapunov functions {which lifts Lyapunov theory to the tangent bundle}.

Here, {we present conditions that establish contraction }with respect to potentially state-dependent, non-Euclidean norms. {The proofs of these conditions rely on the general theoretical framework of Finsler-Lyapunov theory of \cite{Forni:2012qe}, and thus our results are an application of the main results of \cite{Forni:2012qe}. }
To the best of our knowledge, such explicit characterizations of contraction with respect to state-dependent, non-Euclidean norms using matrix measures are not established elsewhere in the literature. 
%

\begin{definition}
  Let $|\cdot|$ be a norm on $\mathbb{R}^n$ and for a {convex} set $\X\subseteq \mathbb{R}^{n}$, let $\Theta:\X\to\mathbb{R}^{n\times n}$ be continuously differentiable and {satisfy $\Theta(x)=\Theta(x)^T\succ 0$ for all $x\in\X$.}
{Let $\mathcal{K}\subseteq \X$ be such that any two points in $\mathcal{K}$ can be connected by a smooth curve,} and, for any two points $x,y\in {\mathcal{K}}$, let {$\Gamma_{\mathcal{K}}(x,y)$}  be the set of piecewise continuously differentiable curves $\gamma:[0,1]\to {\mathcal{K}}$ connecting $x$ to $y$ {within $\mathcal{K}$} so that $\gamma(0)=x$ and $\gamma(1)=y$. The induced \emph{distance metric} is given by
  \begin{align}
    \label{eq:64}
    d_{{\mathcal{K}}}(x,y)=\inf_{\gamma\in\Gamma_{{\mathcal{K}}}(x,y)}\int_{0}^1|\Theta(\gamma(s)){\gamma'}(s)| ds.
  \end{align}
{When $\mathcal{K}=\X$, we drop the subscript and write $d(x,y)$ for $d_\X(x,y)$.}
\end{definition}

\begin{remark}
\label{rem:fin}
{Assuming $\Theta(x)$ is extended to a function defined on an open set containing $\X$}, from a differential geometric perspective, the function $V(x,\delta x):=|\Theta(x)\delta x|$ is a \emph{Finsler function} defined on the tangent bundle {of this open set}, and the distance metric $d(x,y)$ is the \emph{Finsler metric} associated with ${V}$ \cite{Bao:2012mg}. 
\end{remark}

%

In the remainder of this section and in the following section, we study the system
\begin{equation}
  \label{eq:1010}
  \dot{x}=f(t,x)
\end{equation}
for $x\in \Domain\subseteq \mathbb{R}^n$ where {$\mathcal{X}$ is convex}, $f(t,x)$ is differentiable in $x$, and $f(t,x)$ and the Jacobian $J(t,x):= \frac{\partial f}{\partial x}(t,x)$ are continuous in $(t,x)$. {When $\X$ is not an open set, we assume there exists an open set containing $\X$ such that $f(t,\cdot)$ can be extended as a differentiable function on this set and the continuity requirements also hold on this set.}  As before, $\phi(t,x_0)$ denotes the solution to \eqref{eq:1010} at time $t$ when the system is initialized with state $x_0$ at time $t=0$.

{Given \eqref{eq:1010} and continuously differentiable $\Theta:\X\to\mathbb{R}^{n\times n}$, $\dot{\Theta}(t,x)$ is shorthand for the matrix given elementwise by $[\dot{\Theta}(t,x)]_{ij}=\frac{\partial \Theta_{ij}}{\partial x}(x)f(t,x)$. 
}

\begin{prop}
\label{thm:1}
Consider the system \eqref{eq:1010} and suppose $\mathcal{K}\subseteq \Domain$ is forward invariant and {such that any two points in $\mathcal{K}$ can be connected by a smooth curve}. Let $|\cdot|$ be a norm on $\mathbb{R}^n$ with induced matrix measure $\mu(\cdot)$ and let $\Theta:{\mathcal{X}}\to\mathbb{R}^{n\times n}$ be continuously differentiable {and satisfy $\Theta(x)=\Theta(x)^T\succ 0$ for all $x\in\X$.} If there exists $c\geq 0$ such that
\begin{align}
  \label{eq:69}
  \mu\left(\dot{\Theta}(t,x)\Theta(x)^{-1}+\Theta(x)J(t,x)\Theta(x)^{-1}\right) \leq -c
\end{align}
for all $x\in\mathcal{K}$ and all $t\geq 0$, then
\begin{align}
  \label{eq:70}
  d_{{\mathcal{K}}}(\phi(t,y_0),\phi(t,x_0))\leq e^{-ct}d_{{\mathcal{K}}}(y_0,x_0)
\end{align}
for all $x_0,y_0\in \mathcal{K}$. Moreover, if the system is time-invariant so that $\dot{x}=f(x)$, then
\begin{align}
  \label{eq:71}
  |\Theta(\phi(t,x_0))f(\phi(t,x_0))|\leq e^{-ct}|\Theta(x_0)f(x_0)|
\end{align}
for all $x_0\in\mathcal{K}$ and all $t\geq 0$.
\end{prop}

Before proceeding with the proof of Proposition \ref{thm:1}, we first note that when $\Theta(x)\equiv \hat{\Theta}\succ 0$ for some fixed symmetric $\hat\Theta\in\mathbb{R}^{n\times n}$, then $|\cdot|_{\hat{\Theta}}$ defined as $|z|_{\hat{\Theta}}=|\hat\Theta z|$ for all $z$ is a norm and $d(x,y)=|y-x|_{\hat\Theta}$. The corresponding induced matrix measure satisfies $\mu_{\hat{\Theta}}(A)=\mu(\hat{\Theta} A\hat{\Theta}^{-1})$ for all $A$ so that Proposition \ref{thm:1} states: if $\mu_{\hat\Theta}(J(t,x))\leq -c\leq 0$ for all $t$ and $x$, then $|\phi(t,x_0)-\phi(t,y_0)|_{\hat\Theta}\leq e^{-ct}|x_0-y_0|_{\hat\Theta}$ and, in the time-invariant case, $|f(\phi(t,x_0))|_{\hat\Theta}\leq e^{-ct}|f(x_0)|_{\hat\Theta}$, which recovers familiar results for contractive systems with respect to (state-independent) non-Euclidean norms; see, \emph{e.g.}, \cite[Theorem 1]{Sontag:2010fk}, \cite[Proposition 2]{Coogan:2016kx}. Thus, Proposition \ref{thm:1} is an extension of these results to state-dependent, non-Euclidean norms.

\begin{proof}
 Let $V(x,\delta x)=|\Theta(x)\delta x|$. To prove \eqref{eq:70}, we claim that $V(x,\delta x)$ satisfies
\begin{align}
  \label{eq:83}
  \dot{V}(x, \delta x)\leq -cV(x, \delta x)
\end{align}
along trajectories of  the \emph{variational} dynamics
 \begin{align}
   \label{eq:73}
   \dot{x}&=f(t,x)\\
   \label{eq:73-2}\dot{\delta x}&=J(t,x)\delta x
 \end{align}
{when $x(0)\in\mathcal{K}$.} 


To prove the claim, let $(x(t),\delta x(t))$  be some trajectory of the variational dynamics \eqref{eq:73}--\eqref{eq:73-2} {with $x(t)\in\mathcal{K}$ for all $t\geq 0$}. In the following, we omit dependent variables from the notation when clear and write, \emph{e.g.}, $f$, $J$, and $\Theta$ instead of $f(t,x)$, $J(t,x)$, and $\Theta(x)$.  We then have{, for almost all $t$,}
\begin{align}
  \label{eq:66}
&\dot{V}(x,\delta x)\\
&:=\lim_{h\to 0^+}\frac{V(x(t+h),\delta x(t+h))-V(x(t),\delta x(t))}{h}\\
&=\lim_{h\to 0^+}\frac{V(x+hf,\delta x+hJ \delta x)-V(x,\delta x)}{h}\\
&=\lim_{h\to 0^+}\frac{|\Theta(x+hf)(\delta x+hJ \delta x)|-|\Theta(x)\delta x|}{h}\\
&=\lim_{h\to 0^+}\frac{|(\Theta+h\dot{\Theta})(\delta x+hJ\delta x)|-|\Theta \delta x|}{h}\\
&= \lim_{h\to 0^+}\frac{|(I+h(\Theta J\Theta^{-1}+\dot{\Theta}\Theta^{-1}))(\Theta \delta x)|-|\Theta \delta x|}{h}\\
&\leq \lim_{h\to 0^+}\frac{|I+h(\Theta J\Theta^{-1}+\dot{\Theta}\Theta^{-1})||\Theta \delta x|-|\Theta \delta x|}{h}\\
&=\mu\left(\Theta J\Theta^{-1}+\dot{\Theta}\Theta^{-1}\right)|\Theta \delta x|\\
&\leq -c |\Theta(x) \delta x|,
\end{align}
and we have proved \eqref{eq:83}. {Recall that we consider $\mathcal{X}\subseteq \mathcal{X}^o$ for some open set $\X^o$ for which $f(t,\cdot)$ and $\Theta(\cdot)$ can be extended to functions with domain $\X^o$. Using the fact that $\X^o$ is a smooth manifold \cite[pp. 56]{Boothby:1986dz} and $\mathcal{K}$ is a positively invariant, connected subset of $\X^o$, we apply \cite [Theorem 1]{Forni:2012qe} to conclude that \eqref{eq:83} implies \eqref{eq:70} for all $x_0,y_0\in \mathcal{X}$.}

To prove \eqref{eq:71}, now suppose $\dot{x}=f(x)$. Then $\dot{f}(x)=J(x)f(x)$ so that $(x(t),f(x(t)))$ is a trajectory of the variational system for any trajectory $x(t)$ of $\dot{x}=f(x)$. Equation \eqref{eq:71} then follows immediately from \eqref{eq:83}.
\end{proof}

We note that \eqref{eq:71} can be proved directly by applying a change of coordinates. In particular, for $\dot{x}=f(x)$ and $\Theta(x)$ as in the hypotheses of Proposition \ref{thm:1}, consider the time-varying change of variables  $w:=\Theta(x)f(x)$ for which $\dot{w}=\tilde{J}(x) w$ where
\begin{align}
  \label{eq:74}
  \tilde{J}(x):=\dot{\Theta}(x)\Theta^{-1}(x)+\Theta(x)J(x)\Theta^{-1}(x).
\end{align}
By Coppel's Lemma (see, \emph{e.g.}, \cite[Theorem II.8.27]{Desoer:2008bh}), $|w(t)|\leq e^{\int_0^t\mu(\tilde{J}(x(t))} |w(0)|\leq e^{-ct}|w(0)|$ where the second inequality follows from \eqref{eq:69}, providing an alternative interpretation for \eqref{eq:71}.

{
\begin{remark}
From Remark \ref{rem:fin}, $V(x,\delta x)$ given in the proof of Proposition \ref{thm:1} can be regarded as a Finsler function defined on the tangent bundle of {an open set (\emph{i.e.}, manifold) containing } $\Domain$. It follows that $V(x,\delta x)$  is then a \emph{Finsler-Lyapunov} function as defined in \cite{Forni:2012qe} for $\dot{x}=f(t,x)$ where $\delta x$ is the \emph{virtual displacement} associated with the system.
\end{remark}
}

\begin{definition}
 A system for which the hypotheses of Proposition \ref{thm:1} hold with $c<0$ (respectively, $c=0$) is \emph{contractive} (respectively, \emph{nonexpansive}) with respect to $|\cdot|$ and $\Theta(x)$ on $\mathcal{K}$.
\end{definition}

{
\begin{remark}
\label{rem:invariant}
If $\mathcal{K}$ in the statement of Proposition \ref{thm:1} is not forward invariant, then a straightforward modification of the proof implies that the conclusions \eqref{eq:70} and \eqref{eq:71} hold for all $t\geq 0$ such that $\phi(\tau,x_0)\in \mathcal{K}$ and $\phi(\tau,y_0)\in\mathcal{K}$ for all $\tau\in[0,t]$.   
\end{remark}
}


%
%
%
%
%
%
%

\section{A global asymptotic stability result for nonexpansive systems}
\label{sec:glob-asympt-stab}
When the hypotheses of Proposition \ref{thm:1} are only satisfied with $c=0$, the system is nonexpansive as defined above so that the distance between any pair of trajectories is nonincreasing but not necessarily exponentially decreasing as when the system is contractive, \emph{i.e.}, when $c<0$. Nonetheless, if the vector field is periodic and there exists a periodic trajectory that passes through a region in which the contraction property holds locally, then all trajectories entrain, that is, converge, to this periodic trajectory. As a special case, if there exists an open set around an equilibrium in which the contraction property holds locally, then the equilibrium is globally asymptotically stable. We make these statements precise in this section.

\begin{figure}
  \centering
{\footnotesize 
  \begin{tikzpicture} [xscale=3, yscale=2,decoration={
    markings,
    mark=at position 0.25 with {\arrow[line width=1pt]{stealth}}}]
    \draw[black, fill=gray!5, draw=black] (.4,-.5) rectangle (3.1,1.5);
    \draw[gray, fill=gray!40] plot [smooth cycle, tension=.7,rotate around={15:(1.65,.45)}] coordinates { (.8,.25) (1.5,.9) (2,.4) (1.6,-.3)};

    \node[black] at (.5,-.3) {\footnotesize $\Domain$};
    \node[black] at (.85,.8) {\footnotesize $\zeta(t)$};
    \node[black] at (2.4,1.3) {\footnotesize $\mu(\tilde{J}(x,t))\leq 0$};
    \node[black] at (1.55,-.1) {\footnotesize$\mu(\tilde{J}(x,t))\leq -c$};
    \node[black] at (1.6,.15) {\footnotesize $\zeta(t^*)$};
    \node[fill=black, circle, inner sep=1pt] (a) at (1.55,.3) {};
    \node[fill=black, circle, inner sep=1pt] (b) at ($(1.55,.31)+ (24:1.4)$){};
    \draw[dashed, black,postaction={decorate}] plot[smooth cycle, tension=1.4] coordinates  {(a) ($(a)+(0,1)$) ($(a)+(-.6,.5)$) };
    \draw[black, line width=1pt] (a) to[bend right=20pt] node[below,pos=.5]{$\gamma(s)$}(b);
    \node[black] at ($(1.55,.31)+ (24:1.4)+(0,-.1)$){\footnotesize $x$};
    \draw[->, black, line width=.6pt] ($(a)+(0,.08)$) to[bend right=3pt] +($(5:.22)$);
    \draw[->, black, line width=.6pt] ($(a)+(0,.08)+(9:.45)$) to[bend left=3pt] +($(13:-.22)$);
  \end{tikzpicture}
}
  \caption{For a periodic, nonexpansive system with domain $\Domain$, $\mu(\tilde{J}(x,t))\leq 0$ for all $x,t$. If there exists a periodic trajectory $\xi(t)$ and a time $t^*$ such that $\mu(\tilde{J}(t^*,\zeta(t^*)))<0$, then there exists a neighborhood of $\zeta(t^*)$ and an interval of time during which the distance between any other trajectory and the periodic trajectory strictly decreases. It follows that all trajectories must entrain to the periodic trajectory. }
  \label{fig:fig1}
\end{figure}
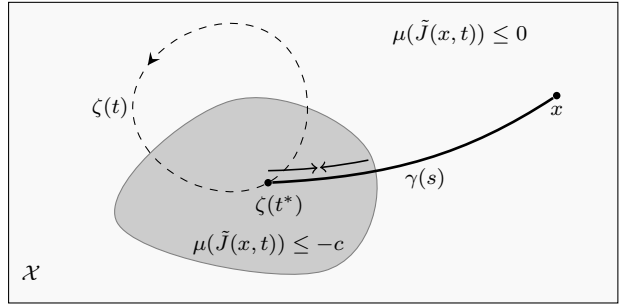

\newcommand{\Ball}{\mathcal{B}}
\begin{thm}
\label{thm:2}
Consider  $\dot{x}=f(t,x)$ for $x\in\Domain\subseteq \mathbb{R}^n$ where $f(t,x)$ is differentiable in $x$, and $f(t,x)$ and the Jacobian $J(t,x):= \frac{\partial f}{\partial x}(t,x)$ are continuous in $(t,x)$. Assume $\Domain$ is forward invariant and {convex}. Let $|\cdot|$ be a norm on $\mathbb{R}^n$ with induced matrix measure $\mu(\cdot)$ and let $\Theta:\X\to\mathbb{R}^{n\times n}$ be continuously differentiable and uniformly positive definite {on $\mathcal{X}$}.


Suppose $f(t,x)$ is $T$-periodic for some $T> 0$ so that $f(t,x)=f(t+T,x)$ for all $t$ and all $x\in\Domain${, and let }$\zeta(t)$ be a {T-}periodic trajectory of the system {so that $\zeta(t)=\zeta(t+T)$ for all $t$}. Define
\begin{align}
  \label{eq:77}
  \tilde{J}(t,x):=\dot{\Theta}(t,x)\Theta(x)^{-1}+\Theta(x)J(t,x)\Theta(x)^{-1}.
\end{align}
If $  \mu(\tilde{J}(t,x))\leq 0  $ for all $x\in\Domain$ and $t\geq 0$, and there exists a time $t^*$ such that
\begin{align}
  \label{eq:76}  \mu(\tilde{J}(t^*,\zeta(t^*)))<0,
\end{align}
then 
\begin{align}
\label{eq:82}  \lim_{t\to\infty} d(\phi(t,x_0),\zeta(t))=0\quad \text{for all $x_0\in\Domain$}
\end{align}
and all trajectories entrain to $\zeta(t)$, that is, 
{
\begin{align}
\label{eq:7}
  \lim_{t\to\infty} |\phi(t,x_0)-\zeta(t)|=0\quad \text{for all $x_0\in\Domain$.}
\end{align}
}
\end{thm}

\begin{proof}
 Without loss of generality, assume $t^*=0$. 
 Then, condition \eqref{eq:76} and continuity of $J(t,x)$, $\Theta(x)^{-1}$, and $\dot{\Theta}(t,x)$ imply there exists $\epsilon>0$, $c>0$, and $0<\tau\leq T$ such that
  \begin{align}
    \label{eq:600}
\mu(\tilde{J}(t,y))\leq -c   \qquad \forall t\in[0,\tau], \ \forall y\in B_{2\epsilon}(\zeta(t))
  \end{align}
where $\Ball_{2\epsilon}(y)=\{z\in\X:d(y,z){< 2\epsilon}\}$.   {Furthermore, with $\mathcal{K}:=\Ball_{2\epsilon}(y)$, note that $d(y,z)=d_{\mathcal{K}}(y,z)$ for all $z\in\Ball_{2\epsilon}(y)$.}

Define the mapping
\begin{align}
  \label{eq:108}
  P(\xi)=\phi(T,\xi)
\end{align}
and observe that $P^k(\xi)=\phi(kT,\xi)$. Let $\zeta^*=\zeta(0)$ and note that $\zeta^*$ is a fixed point of $P$. 

{First, }consider a point $\xi\in{\bar\Ball}_\epsilon(\zeta^*){:=\{z\in\mathcal{X}:d(y,z)\leq \epsilon\}}$. Let $x(t)=\phi(t,\xi)$ and note that $d(\zeta^* ,P(\xi))=d(\zeta(T) ,x(T))$.
We have $ d(\zeta(T),x(T))\leq d(\zeta(\tau), x(\tau))$ since $  \mu(\tilde{J}(t,x))\leq 0  $ for all $x\in\Domain$ and $t\geq 0$,
and, 
{by Remark \ref{rem:invariant} with $\mathcal{K}=B_{2\epsilon}(y)$,} \eqref{eq:600} implies
\begin{align}
  \label{eq:1017}
d(\zeta(\tau) ,x(\tau))\leq e^{-c\tau}d(\zeta(0) ,x(0)).  
\end{align}

Now consider $\xi\in \Domain$ such that $d(\zeta^*,\xi)> \epsilon$ and again let $x(t)=\phi(t,\xi)$.  Let $\delta=(1-e^{-c\tau})\epsilon/2$ and let $\gamma:[0,1]\to \Domain$ with $\gamma(0)=\zeta^*$ and $\gamma(1)=\xi$ be such that $\int_0^1|\Theta(\gamma(s))\gamma'(s)|ds\leq d(\zeta^*,\xi)+\delta$. Since $d(\zeta^*,\gamma(0))=0$ and $d(\zeta^*,\gamma(1))>\epsilon$, by continuity of $\gamma(s)$ and $d(\zeta^*,\cdot)$, there exists $s_\epsilon$ such that $d(\zeta^*,\gamma(s_\epsilon))=\epsilon$. Moreover,
\begin{align}
  \label{eq:78}
&  d(\zeta^*,\gamma(s_\epsilon))+d(\gamma(s_\epsilon),\xi)\\
&\leq \int_0^{s_\epsilon}|\Theta(\gamma(s))\gamma'(s)|ds+\int_{s_\epsilon}^{1}|\Theta(\gamma(s))\gamma'(s)|ds\\
&\leq d(\zeta^*,\xi)+\delta
\end{align}
where the first inequality follows from the definition of $d$ and the fact that arc-length is independent of reparameterizations of $\gamma$.

Let $\sigma(t)=\phi(t,\gamma(s_\epsilon))$.   By Proposition \ref{thm:1}, $\mu(\tilde{J}(t,x))\leq 0$ for all $x\in\Domain$ and  all $t\geq 0$ implies $d(\phi(t,x_0),\phi(t,y_0))\leq d(x_0,y_0)$ for all $x_0,y_0\in\Domain$. In particular, $d(\sigma(T) ,x(T))\leq d(\sigma(0),x(0))= d(\gamma(s_\epsilon),\xi)\leq d(\zeta^*,\xi)+\delta-\epsilon$. Furthermore, by the same argument as in the preceding case, we have $d(\zeta(T),\sigma(T))\leq e^{-c\tau}d(\zeta^*,\gamma(s_\epsilon))=e^{-c\tau}\epsilon$. Thus, by the triangle inequality,
\begin{align}
  \label{eq:1012}
  d(\zeta^*,P(\xi))&\leq d(\zeta(T),\sigma(T)) +d(\sigma(T) ,x(T))\\
&\leq d(\zeta^*,\xi)+\delta-(1-e^{-c\tau})\epsilon\\
&=d(\zeta^*,\xi)-\delta.
\end{align}
We thus have
\begin{align}
  \label{eq:1050}
  d(\zeta^*,P(\xi))\leq 
  \begin{cases}
    d(\zeta^*,\xi)-\delta&\text{if }d(\zeta^*,\xi)>\epsilon\\
    e^{-c\tau}d(\zeta^*,\xi)&\text{if }d(\zeta^*,\xi)\leq \epsilon.
  \end{cases}
\end{align}
It follows that, for all $\xi$, $d(\zeta^*,P^k(\xi))\leq \epsilon$ for some finite $k$ (in particular, for any $k\geq d(\zeta^*,\xi)/\delta$). The second condition of \eqref{eq:1050}  ensures $d(\zeta^*,P^k(\xi))\to 0$ as $k\to \infty$ so that \eqref{eq:82} holds. 
{Moreover, since $\Theta(x)$ is uniformly positive definite on $\X$ by hypothesis, there exists a constant $m>0$ such that $|\Theta(x)z|\geq m|z|$ for any $x, z$. It follows that $d(x,y)\geq \inf_{\gamma\in\Gamma(x,y)}\int_{0}^1m|{\gamma'}(s)| ds=m|y-x|$ where the equality holds because $\X$ is convex. Thus, \eqref{eq:82} implies \eqref{eq:7}.}
\end{proof}

Theorem \ref{thm:2} and its proof are illustrated in Figure \ref{fig:fig1}.

Note that, as a consequence of the global entrainment property, {$\zeta(t)$} in the statement of Theorem \ref{thm:2} is the unique periodic trajectory of a system satisfying the hypotheses of the theorem.

Theorem \ref{thm:2} is closely related to existing results in the literature, although we believe the generality provided by Theorem \ref{thm:2} is novel. In particular, \cite[Lemma 6]{Lovisari:2014yq} provides a similar result for $\mu(\cdot)$ restricted to the matrix measure induced by the $\ell_1$ norm under the assumption that \eqref{eq:1} is monotone and time-invariant. A similar technique is applied to periodic trajectories of a class of monotone flow networks in \cite[Proposition 2]{Lovisari:2014qv}, but a general formulation is not presented. 

While Theorem \ref{thm:2} is interesting in its own right, in this paper, {we are primarily interested in obtaining Lyapunov functions for monotone systems. Thus,} our main interest is in the following Corollary, which specializes Theorem \ref{thm:2} to time-invariant systems. 

\begin{cor}
\label{cor:nonexpeq}
Consider  $\dot{x}=f(x)$ for $x\in\Domain\subseteq \mathbb{R}^n$ for continuously differentiable $f(x)$. Assume $\Domain$ is forward invariant and {convex}. Let $|\cdot|$ be a norm on $\mathbb{R}^n$ with induced matrix measure $\mu(\cdot)$ and let $\Theta:\X\to\mathbb{R}^{n\times n}$ be continuously differentiable and uniformly positive definite {on $\X$}.  Let $x^*$ be an equilibrium of the system and define
\begin{align}
  \label{eq:77}
  \tilde{J}(x):=\dot{\Theta}(x)\Theta(x)^{-1}+\Theta(x)J(x)\Theta(x)^{-1}
\end{align}
where $J(x)=\frac{\partial f}{\partial x}(x)$. If
\begin{align}
  \label{eq:79}
  \mu(\tilde{J}(x))&\leq 0 \quad \text{for all }x\in\Domain,\text{ and}\\
  \label{eq:79-2}  \mu(\tilde{J}(x^*))&<0,
\end{align}
then $x^*$ is unique and globally asymptotically stable. Moreover,
  $d(x,x^*)$
is a global Lyapunov function and 
$  |\Theta(x)f(x)|$
is a local Lyapunov function. If $  |\Theta(x)f(x)|$ is globally proper, then it is also a global Lyapunov function.%
\end{cor}
\begin{proof}
  Choose any $T>0$, for which $f(x)$ is then (vacuously) $T$-periodic, and $\zeta(t):=x^*$ is trivially a periodic trajectory so that Theorem \ref{thm:2} applies. Note that we may take $\tau=T$ in the proof of the theorem. {Therefore,  $\lim_{t\to\infty}d(\phi(t,x_0),x^*)= 0$ for all $x_0\in\X$. Moreover, as in the proof of Theorem \ref{thm:2}, there exists $m>0$ such that $d(x,y)\geq m|y-x|$ because $\Theta(x)$ is uniformly positive definite and $\mathcal{X}$ is convex so that $d(x,x^*)$ is globally proper. 

}

{We now show that $|\Theta(x)f(x)|$ is proper. Since $\Theta(x)$ is uniformly positive definite, 
it is sufficient to show $|f(x)|$ is proper. Let $f(x)=J(x^*)(x-x^*)+g(x)$ for appropriately defined $g$ satisfying that for any $\epsilon_1>0$, there exists $\epsilon_2>0$ such that $|g(x)|<\epsilon_1 |x-x^*|$ for all $x$ satisfying $|x-x^*|<\epsilon_2$. Moreover, \eqref{eq:79-2} implies $\tilde{J}(x^*)$ is Hurwitz and since $\tilde{J}(x^*)$ is related to $J(x^*)$ via a similarity transform, it follows that $J(x^*)$ is Hurwitz and thus invertible. Therefore, there exists $\epsilon_3>0$ for which $|J(x^*)(x-x^*)|\geq \epsilon_3 |x-x^*|$ for all $x\in\X$. Then, $|f(x)|\geq |J(x^*)(x-x^*)|-|g(x)|\geq \epsilon_3|x-x^*|-|g(x)|$ for all $x\in\X$. Choosing $\epsilon_1=\epsilon_3/2$ implies that $|f(x)|\geq \frac{\epsilon_3}{2}|x-x^*|$ for all $x\in\X$ satisfying $|x-x^*|<\epsilon_2$, and thus $|f(x)|$ is proper.}


From Proposition \ref{thm:1} with $c=0$, we have that $d(x,x^*)$ and $|\Theta(x)f(x)|$ are nonincreasing by \eqref{eq:70} and \eqref{eq:71}, and both tend to zero along trajectories of the system, thus they both satisfy condition (iii) of Definition \ref{def:lyap}. {Observing that $d(x,x^*)$ and $|\Theta(x)f(x)|$ are both positive definite completes the proof, where we remark that global asymptotic stability of $x^*$ follows from the existence of a global Lyapunov function.}
\end{proof}

\section{Proof of main result}
\label{sec:proof-main-result}
We are now in a position to prove our main results, Theorems \ref{thm:tvmain1} and \ref{thm:tvmain2}. We begin with the following proposition, which shows that $d(x,y)$ as in \eqref{eq:64} can be obtained explicitly when $\Theta(x)$ is a diagonal, state-dependent weighting matrix.

\begin{prop}
Let $\X\subseteq \mathbb{R}^n$ be a rectangular set and for any $x,y\in\mathcal{X}$, let $\Gamma(x,y)$ be the set of piecewise continuously differentiable curves $\gamma:[0,1]\to \mathcal{X}$ connecting $x$ to $y$ {within $\X$} so that $\gamma(0)=x$ and $\gamma(1)=y$.  Consider $\Theta:\X\to \mathbb{R}^{n\times n}$ given by $\Theta(x)=\textnormal{diag}\{\theta_1(x_1),\ldots,\theta_n(x_n)\}$ where $\{\theta_i(\cdot)\}_{i=1}^n$ is a collection of nonnegative continuously differentiable functions. Suppose there exists $c>0$ such that $\theta_i(x_i)>c$ for all $x\in\mathcal{X}$ for all $i=1,\ldots,n$. Let $|\cdot|$ be a norm on $\mathbb{R}^n$ and consider the metric $d(x,y)$ given by \eqref{eq:64}. 
Then
\begin{align}
  \label{eq:65}
d(x,y)=
\left|  \begin{bmatrix}
    \int_{x_1}^{y_1}\theta_i(\sigma_1)d\sigma_1&\ldots&    \int_{x_n}^{y_n}\theta_i(\sigma_n)d\sigma_n
  \end{bmatrix}^T\right|.
\end{align}
\end{prop}
\begin{proof}
Given any curve $\gamma\in\Gamma(x,y)$, define a new curve $\tilde{\gamma}(s):[0,1]\to\mathbb{R}^n$ as follows. Let $\tilde{\gamma} (0)=x$ and let ${\tilde{\gamma}_i}'(s)=\theta_i(\gamma_i(s)){\gamma'_i}(s)$ for all $i$ so that
\begin{align}
  \label{eq:68}
  \int_0^1|\Theta(\gamma(s)){\gamma'}(s)| ds=\int_0^1|\tilde{\gamma}'(s)| ds.
\end{align}
{That is, if we understand the left-hand side of \eqref{eq:68} to be the length of the curve $\gamma$ computed using the metric of \eqref{eq:64}, then this length is equal to the standard arc-length of $\tilde{\gamma}$ using the metric induced by the norm $|\cdot|$.} 
Observe that, for all $i$,
\begin{align}
  \label{eq:67}
 \tilde{\gamma}_i(1)-x_i&=\int_0^1\theta_i(\gamma_i(s)){\gamma'_i}(s)ds=\int_{x_i}^{y_i}\theta_i(\sigma )d\sigma.
\end{align}
Then $z:=\tilde{\gamma}(1)$ is a point depending only on $x$ and $y$ and is independent of the particular curve $\gamma$. {That is, for every $x,y\in\mathcal{X}$, there exists a unique $z\in\mathbb{R}^n$ such that each curve $\gamma$ connecting $x$ to $y$ generates a curve $\tilde{\gamma}$ as defined above connecting $x$ to $z$. Observe that}
\begin{align}
  \label{eq:80}
  \inf_{\tilde{\gamma}\in\Gamma(x,z)}\int_{0}^1\left|\begin{bmatrix}{\tilde\gamma'_1}(s)&\ldots&{\tilde\gamma'_n}(s)\end{bmatrix}^T\right| ds&=|z-x|
\end{align}
{where the infimum is achieved for} $\tilde{\gamma}(s)=(1-s)x+sz${, and the equality follows since} $\tilde{\gamma}'(s)=z-x$. {Furthermore, this particular minimizing $\tilde{\gamma}$ is generated from the unique curve} $\gamma\in\Gamma(x,y)$ satisfying the decoupled system of differential equations $z_i-x_i=\theta_i(\gamma_i(s))\gamma_i'(s)$ for all $i$. Moreover, $\gamma_i(s)$ is monotonic for all $i$ and all $s\in[0,1]$ since $\gamma_i'(s)$ does not change sign, and thus $\gamma(s)$ is contained in the rectangle defined by the corners $x$ and $y$ so that $\gamma(s)\in\Domain$ for $s\in[0,1]$. It follows that $d(x,y)=|z-x|$, which is equivalent to \eqref{eq:65}.
\end{proof}

\textbf{Proof of Theorem \ref{thm:tvmain1}.}
Let 
\begin{align}
  \label{eq:85}
  \Theta(x):=\textnormal{diag}\{\theta_1(x_1),\ldots,\theta_n(x_n)\}
\end{align}
and consider the $\ell_1$ norm $|\cdot|_1$ with induced matrix measure $\mu_1(\cdot)$. Define $\tilde{J}(x)$ as in \eqref{eq:77}. %
Recall ${J}(x)=\frac{\partial f}{\partial x}(x)$ is Metzler since the system is monotone. Then also $\tilde{J}(x)$ is Metzler since $\Theta(x)$ is diagonal and uniformly positive definite so that $\Theta(x)J(x)\Theta(x)^{-1}$ retains the sign structure of $J(x)$ and $\dot{\Theta}(x)\Theta(x)^{-1}$ is a diagonal matrix.

Then
\begin{align}
  \label{eq:75}
  \mu_1(\tilde{J}(x))%
&=\max_{j=1\ldots,n}\left(\sum_{i=1}^n\tilde{J}_{ij}(x)\right)\\
  \label{eq:75-2}&=\max_{j=1\ldots,n}\Bigg(\Bigg(\dot{\theta}_j(x) +\sum_{i=1}^n\theta_{i}(x){J}_{ij}(x)\Bigg) \theta_j(x)^{-1}\Bigg)
\end{align}
where the first equality follows by \eqref{eq:10} since $\tilde{J}(x)$ is Metzler. It follows from \eqref{eq:111} and \eqref{eq:75-2} that $\mu_1(\tilde{J}(x))\leq 0$ for all $x\in\Domain$, {and it follows from \eqref{eq:114} and \eqref{eq:75-2} that $\mu_1(\tilde{J}(x^*))<0$ because $\dot{\theta}_j(x^*)=0$ and $\theta_j(x_j^*)^{-1}>0$ for all $j$.}
Applying Corollary \ref{cor:nonexpeq} establishes that \eqref{eq:113} is a global Lyapunov function and \eqref{eq:113-2} is a local Lyapunov function.
\qed 

\textbf{Proof of Theorem \ref{thm:tvmain2}.}
Let $\theta_i(x):=1/\omega_i(x)$ and define
\begin{align}
  \label{eq:87}
  \Omega(x)&:=\textnormal{diag}\{\omega_1(x_1),\ldots,\omega_n(x_n)\}\\
\Theta(x)&:=\textnormal{diag}\{\theta_1(x_1),\ldots,\theta_n(x_n)\}.
\end{align}
Since $0< \omega_i(x)\leq c$ for all $x\in\mathcal{X}$ for $i=1,\ldots,n$, we have $0<c^{-1}\leq \theta_i(x)$ for all $x\in\mathcal{X}$ for $i=1,\ldots,n$. 

Observe that %
\begin{align}
  \label{eq:88}
\dot{\theta}_i(x)\theta_i(x_i)^{-1}=\frac{\frac{d\theta_i}{d x_i}(x_i)}{\theta_i(x_i)}f_i(x)&=\frac{-\frac{d\omega_i}{d x_i}(x_i)}{\omega_i(x_i)}f_i(x)\\
  \label{eq:88-2}&=-\dot{\omega}_i(x)\omega_i(x_i)^{-1}
\end{align}
for all $i=1,\ldots,n$ where the second equality is established from the identity
\begin{align}
  \label{eq:6}
  \frac{d\omega_i}{d x_i}(x_i)=\frac{d}{d x_i}\left(\frac{1}{\theta_i(x_i)}\right)=\frac{-(d\theta_i/dx_i)(x_i)}{(\theta_i(x_i))^2}.
\end{align}
As in the proof of Theorem {\ref{thm:tvmain1}}, define  $\tilde{J}(x)$ according to \eqref{eq:77},
and now consider the $\ell_\infty$ norm $|\cdot|_\infty$ with induced matrix measure $\mu_\infty(\cdot)$.  As before, $\tilde{J}(x)$ is Metzler so that
\begin{align}
  \label{eq:91}
&  \mu_{\infty}(\tilde{J}(x))=\max_{i=1,\ldots,n}\sum_{j=1}^n\tilde{J}_{ij}(x)\\
&=\max_{i=1,\ldots,n}\Bigg(\dot{\theta}_i(x)\theta_i(x)^{-1}+\theta_i(x_i)\sum_{j=1}^nJ_{ij}(x)\theta_j(x_j)^{-1}\Bigg)\\
  \label{eq:91-3}&=\max_{i=1,\ldots,n}\Bigg(\omega_i(x_i)^{-1}\Bigg(-\dot{\omega}_i(x)+\sum_{j=1}^nJ_{ij}(x)\omega_j(x_j)\Bigg)\Bigg)
\end{align}
where the first equality follows from \eqref{eq:10-2} and the last equality follows from \eqref{eq:88}--\eqref{eq:88-2}. It follows from \eqref{eq:45} and \eqref{eq:91-3} that $\mu_\infty(\tilde{J}(x))\leq 0$ for all $x\in\Domain$, and {it follows from \eqref{eq:48} and \eqref{eq:91-3} that $\mu_\infty(\tilde{J}(x^*))<0$ because $\dot{\omega}_i(x^*)=0$ and $\omega_i(x_i^*)^{-1}>0$ for all $i$.}
 Applying Corollary \ref{cor:nonexpeq} establishes that \eqref{eq:47} is a global Lyapunov function and \eqref{eq:47-2} is a local Lyapunov function. %
\qed

\textbf{Proofs of Corollaries \ref{cor:1} and \ref{cor:2}.}
{We need only show that \eqref{eq:17-2} (respectively, \eqref{eq:20}) is a global Lyapunov function when the stated conditions hold, as the remainder of the claims follow immediately from Theorems \ref{thm:tvmain1} and \ref{thm:tvmain2}.  Below, we show that \eqref{eq:17-2} is globally proper in this case. A symmetric argument establishes that \eqref{eq:20} is globally proper.}



%

To this end, suppose {$v^TJ(x)\leq -c\mathbf{1}^T$ for all $x\in\mathcal{X}$ for some $c>0$.} Then there also exists some $\tilde{c}>0$ such that  $v^TJ(x)\leq -\tilde{c} v$ for all $x\in \mathcal{X}$. In particular, we take $0<\tilde{c} \leq c/|v|_\infty$.  With $\Theta:=\text{diag}\{v\}$, it follows that $\mu_1(\Theta J(x)\Theta^{-1})\leq -\tilde{c}$ for all $x\in\Domain$. 

From a slight modification of \cite[Theorem 33, pp. 34--35]{Desoer:2008bh}, we then have that $|\Theta f(x)|_1\geq \tilde{c} |\Theta x|_1$, which implies that $|\Theta f(x)|_1$ is globally proper. For completeness, we repeat this argument here. Since $\Theta f(x)=\int_0^1\Theta J(s x) x \ ds$,
\begin{align}
  \label{eq:106}
  |\Theta f(x)|_1&\geq -\mu_1\left(\int_0^1\Theta J(s x )\Theta^{-1}  \ ds\right) \left|\Theta x\right|_1\\
&\geq -\left(\int_0^1\mu_1(\Theta J(s x )\Theta^{-1})ds \right)|\Theta x|_1\\
&\geq \tilde{c} |\Theta x|_1
\end{align}
 where the first inequality follows from the fact that $|Ax|\geq -\mu(A)|x|$ for all $A$ and all $x$ where $|\cdot|$ and $\mu(\cdot)$ are any vector norm and corresponding matrix measure, and the second inequality follows from the fact that $\mu(A+B)\leq \mu(A)+\mu(B)$ for all $A$, $B$ (see \cite{Desoer:2008bh} for a proof of both these facts). Since \eqref{eq:17-2} is equivalent to $|\Theta f(x)|_1$, the proof is complete.  
\qed

\section{An algorithm for computing separable Lyapunov functions}
\label{sec:an-algor-comp}
In this section, we briefly discuss an efficient and scalable algorithm for computing $\theta(x)$ and $\omega(x)$ in Theorems \ref{thm:tvmain1} and \ref{thm:tvmain2} when each element of $f(x)$ is a polynomial or rational function of $x$. %
Thus, the proposed approach provides an efficient means for computing sum-separable and max-separable Lyapunov functions of monotone systems.

Our proposed algorithm relies on \emph{sum-of-squares (SOS) programming} \cite{Parrilo:2000fk, Parrilo:2001uq}. A polynomial $s(x)$ is a \emph{sum-of-squares polynomial} if $s(x)=\sum_{i=1}^r(g_i(x))^2$ for some polynomials $g_i(x)$ for $i=1,\ldots,r$. A \emph{SOS feasibility problem} consists in finding a collection of polynomials $p_i(x)$ for $i=1,\ldots,N$ and a collection of SOS polynomials $s_i(x)$ for $i=1,\ldots,M$ such that 
\begin{align}
  \label{eq:96}
  a_{0,j}+\sum_{i=1}^Np_i(x)a_{i,j}(x)+\sum_{i=1}^Ms_i(x)b_{i,j}(x)&\\
  \label{eq:96-2}\qquad\text{are SOS polynomials for $j=1,\ldots,J$}&
\end{align}
for fixed polynomials $a_{0,j}(x)$, $a_{i,j}(x)$, and $b_{i,j}(x)$ for all $i,j$. The set of polynomials $\{p_i(x)\}_{i=1}^N$ and $\{s_i(x)\}_{i=1}^M$ satisfying \eqref{eq:96}--\eqref{eq:96-2} forms a convex set, and by fixing the degree of these polynomials, we arrive at a finite dimensional convex optimization problem. There exists efficient computational toolboxes that convert SOS feasibility programs into standard semi-definite programs (SDP) \cite{sostools}. SOS programming has led to efficient computational algorithms for a number of controls-related applications such as searching for polynomial Lyapunov functions \cite{Papachristodoulou:2002jk}, underapproximating regions of attraction \cite{Topcu:2007fk}, and safety verification \cite{Coogan:2015dq}.

Here, we present a SOS feasibility problem that is sufficient for finding $\theta(x)$ or $\omega(x)$ satisfying Theorem \ref{thm:tvmain1} or \ref{thm:tvmain2}. To that end, recall that in the hypotheses of these theorems, we assume $\Domain$ is rectangular. In this section, for simplicity, we assume $\X$ is a closed set with nonzero measure so that $\X=\mathbf{cl}\{x: a_i< x_i< b_i \text{ for all $i$}\}$ for appropriately defined $a_i$ and $b_i$ where possibly $a_i=-\infty$ and/or $b_i=\infty$  and where $\mathbf{cl}$ denotes closure. Equivalently, $\X=\{x:d(x)\geq 0\}$ where $d(x)=\begin{bmatrix}d_1(x_1)&\ldots&d_n(x_n)\end{bmatrix}^T$ and
\begin{align}
  \label{eq:92}
  d_i(x_i)=
  \begin{cases}
    x_i-a_i&\text{if $a_i\neq -\infty$ and $b_i=\infty$}\\
    (x_i-a_i)(b_i-x_i)&\text{if $a_i\neq -\infty$ and $b_i\neq \infty$}\\
    b_i-x_i&\text{if $a_i= -\infty$ and $b_i\neq \infty$}\\
    0&\text{if $a_i= -\infty$ and $b_i= \infty$}.
  \end{cases}
\end{align}

If $\sigma(x)=\begin{bmatrix}\sigma_{1}(x) & \ldots & \sigma_n(x)\end{bmatrix}^T$ where each $\sigma_j(x)$, $j=1\ldots n$ is a SOS polynomial, we call $\sigma(x)$ a \emph{SOS $n$-vector}. 

\begin{prop}
  Let \eqref{eq:1} be a monotone system with equilibrium $x^*$ and rectangular domain $\mathcal{X}=\{x:d(x)\geq 0\}$ where $d(x)=\begin{bmatrix}d_1(x_1)&\ldots&d_n(x_n)\end{bmatrix}^T$. Suppose each $f_i(x)$ is polynomial. Then the following is a SOS feasibility problem, and a feasible solution provides $\theta(x)$ satisfying the conditions of Theorem \ref{thm:tvmain1}:
{
\begin{alignat}{2}
\nonumber &\textnormal{For some fixed $\epsilon>0$,}\\
\nonumber &\textnormal{Find:}\\
\nonumber &\quad\textnormal{Polynomials $\theta_i(x_i)$, for $i=1,\ldots,n$} \\
\nonumber &\quad\textnormal{SOS polynomials $s_i(x_i)$ for $i=1,\ldots,n$}\\
\nonumber &\quad\textnormal{SOS $n$-vectors $\sigma^{i}(x)$ for $i=1,\ldots,n$}\\
\nonumber &\textnormal{Such that:}\\
\nonumber &\quad (\theta_i(x_i)-\epsilon)-s_i(x_i)d_i(x_i) \quad &&\hspace*{-.5in}\textnormal{is a SOS polynomial}\\
  \label{eq:97-3}&&&\hspace*{-.5in}\textnormal{for $i=1,\ldots,n$}\\
\nonumber &\quad -(\theta(x)^TJ(x)+\dot{\theta}(x)^T)_i-\sigma^i(x)^Td(x) \\
\nonumber & &&\hspace*{-.5in}\textnormal{is a SOS polynomial}\\
  \label{eq:97-4}& &&\hspace*{-.5in}\textnormal{for $i=1,\ldots,n$}\\
  \label{eq:97-2}&\quad -\theta(x^*)^TJ(x^*)-\epsilon\geq 0
\end{alignat}
}
where $(\theta(x)^TJ(x)+\dot{\theta}(x)^T)_i$ denotes the $i$-th entry of $\theta(x)^TJ(x)+\dot{\theta}(x)^T$.
\end{prop}
\begin{proof}
Note that $(\theta(x)^TJ(x)+\dot{\theta}(x)^T)_i$ is a polynomial in $x$ for which the decision variables of the SOS feasibility problem %
appear linearly. In addition, \eqref{eq:97-2} is a linear constraint on the coefficients of $\theta(x)$. Thus, \eqref{eq:97-3}--\eqref{eq:97-2} is a well-defined SOS feasibility problem.

Next, we claim that \eqref{eq:97-3} is sufficient for $\theta(x)\geq \epsilon \mathbf{1}$ for all $x\in\Domain$. To  prove the claim, suppose \eqref{eq:97-3} is satisfied so that $(\theta_i(x_i)-\epsilon)-s_i(x_i)d_i(x_i)$ is a SOS polynomial and consider $x\in \Domain$ so that $d(x)\geq 0$. Since $s_i(x_i)$ is a SOS polynomial for all $i$, we have that also $s_i(x_i)d(x_i)\geq 0$ so that $\theta_i(x_i)\geq \epsilon + s_i(x_i)d(x_i)\geq \epsilon$.

A similar argument shows that if \eqref{eq:97-4} holds for all $i$, then condition (1) of Theorem \ref{thm:tvmain1} holds, that is, \eqref{eq:111} holds for all $x\in \Domain$. Indeed, suppose $-(\theta(x)^TJ(x)+\dot{\theta}(x)^T)_i-\sigma^i(x)^Td(x)$ is a SOS polynomial and consider $x\in\Domain$ so that $d(x)\geq 0$. Since $\sigma^i(x)$ is a SOS $n$-vector, $\sigma^i(x)^Td(x)\geq 0$ and thus $-(\theta(x)^TJ(x)+\dot{\theta}(x)^T)_i\geq 0$.

Finally, \eqref{eq:97-2} implies \eqref{eq:114}.
\end{proof}

The technique employed in \eqref{eq:97-3} and \eqref{eq:97-4} to ensure that the conditions of Theorem \ref{thm:tvmain1} hold whenever $d(x)\geq 0$ is similar to the $\mathcal{S}$-procedure used to express constraints on quadratic forms as linear matrix inequalities \cite [p. 23]{Boyd:1994uq} and is common in applications of SOS programming to systems and control theory.

\begin{remark}
  A symmetric proposition holds establishing a SOS feasibility program sufficient for computing $\omega(x)$ that satisfies the conditions of Theorem \ref{thm:tvmain2}. We omit an explicit form for this SOS program due to space constraints.
\end{remark}

\addtocounter{example}{-2}

\begin{example}[Cont.]
  We consider the system from Example \ref{ex:statedep} and seek to compute a sum-separable Lyapunov function using the SOS program of Section \ref{sec:an-algor-comp}. We let $\epsilon=0.01$ and search for $\theta_1(x_1)$ and $\theta_2(x_2)$ that are polynomials of up to second-order. We consider SOS polynomials that are 0-th order, that is, all SOS polynomials are considered to be positive constants, which proves sufficient for this example. The SOS program requires 0.26 seconds of computation time and returns
  \begin{align}
    \label{eq:98}
    \theta_1(x)&=1.7429\\
     \theta_2(x)&=x_2^2 + 1.3793x_2 + 1.9503
  \end{align}
where parameters have been scaled so that the leading coefficient of $\theta_2(x_2)$ is $1$ since scaling $\theta(x)$ does not affect the validity of the resulting SOS program or the conditions of Theorem \ref{thm:tvmain1}. Then, \eqref{eq:113} and \eqref{eq:113-2} give the following sum-separable Lyapunov functions:
\begin{align}
  \label{eq:44}
  V(x)&=1.7429x_1+\frac{1}{3}x_2^3+\frac{1.3793}{2}x_2^2+1.9503 x_2\\
V(x)&=1.7429 \left|{-x_1}+x_2^2\right| + x_2^3 + 1.3793x_2^2 + 1.9503 x_2.
\end{align}
$\blacksquare$
\end{example}
\addtocounter{example}{2}

\section{Applications}
\label{sec:examples}
In this section, we present several applications of our main results. First, we establish a technical result that will be useful for constructing Lyapunov functions as the limit of a sequence of contraction metrics.

\begin{prop}
\label{prop:limit}
  Let $x^*\in \Domain$ be an equilibrium for \eqref{eq:1}.  Suppose there exists a sequence of global Lyapunov functions $V^i:\Domain\to\mathbb{R}_{\geq 0}$ for \eqref{eq:1} that converges locally uniformly to $V(x):=\lim_{i\to\infty}V^i(x)$. If $V(x)$ is positive definite and globally proper (see Definition \ref{def:lyap}) then $V(x)$ is also a global Lyapunov function for \eqref{eq:1}.

\end{prop}

\begin{proof}
Note that $x^*$ is globally asymptotic stable since there exists a global Lyapunov function, and consider  some $x_0\in \Domain$. Asymptotic stability of $x^*$ implies there exists a bounded set $\Omega$ for which $\phi(t,x_0)\in\Omega\subseteq \Domain$ for all $t\geq 0$. For $i=1,\ldots,n$, we have $V^i(\phi(t,x_0))$ is nonincreasing in $t$ and $\lim_{t\to\infty}V^i(\phi(t,x_0))=0$. Local uniform convergence establishes $V(x)$ is continuous, $V(\phi(t,x_0))$ is nonincreasing in $t$, and $\lim_{t\to\infty}V(\phi(t,x_0))=0$, and thus condition (3) of Definition \ref{def:lyap} holds. Under the additional hypotheses of the proposition, we have that $V(x)$ is therefore a global Lyapunov function.%
\end{proof}

Note that a sequence $V^i(x)$ arising from a sequence of weighted contraction metrics, \emph{i.e.}, $V^i(x)=|P_i(x-x^*)|$ or $V^i(x)=|P_if(x)|$ for $P_i$ converging to some nonsingular $P$, satisfies the conditions of Proposition \ref{prop:limit}.

The following example is inspired by \cite [Example 3]{Dirr:2015rt}.
\begin{example}[Comparison system]
  Consider the system 
  \begin{align}
    \label{eq:23}
    \dot{x}_1&=-x_1+x_1x_2\\
    \label{eq:23-2}\dot{x}_2&=-2x_2-x_2^2+\gamma(x_1)^2
  \end{align}
evolving on $\Domain=\mathbb{R}_{\geq 0}^2$ where $\gamma:\mathbb{R}_{\geq 0}\to\mathbb{R}_{\geq 0}$ is strictly increasing  and satisfies $\gamma(0)=0$, $\bar{\gamma}:=\lim_{\sigma\to\infty}\gamma(\sigma)<1$, and $  \gamma'(\sigma)\leq \frac{1}{(1+\sigma)^{2}}$.
Consider the change of coordinates $(\eta_1,\eta_2)=(\log(1+x_1),x_2)$ so that 
\begin{align}
  \label{eq:27}
  \dot{\eta}_1&=\frac{1}{1+x_1}(-x_1+x_1x_2)
\end{align}
where we substitute $(x_1,x_2)=(e^{\eta_1}-1,\eta_2)$. Then
\begin{align}
  \label{eq:28}
  \dot{\eta}_1\leq -\beta(e^{\eta_1}-1)+\eta_2
\end{align}
where $\beta(\sigma)={\sigma}/{(1+\sigma)}$. Introduce the comparison system 
\begin{align}
  \label{eq:29}
  \dot{\xi}_1&=-\beta(e^{\xi_1}-1)+\xi_2\\
  \label{eq:29-2}  \dot{\xi}_2&=-2\xi_2-\xi_2^2+\gamma(e^{\xi_1}-1)^2
\end{align}
evolving on $\mathbb{R}_{\geq 0}^2$. The comparison principle (see, \emph{e.g.}, \cite{Dirr:2015rt}) ensures that asymptotic stability of the origin for the comparison system \eqref{eq:29}--\eqref{eq:29-2} implies asymptotic stability of the origin of the $(\eta_1,\eta_2)$ system, which in turn establishes asymptotic stability of the origin for \eqref{eq:23}--\eqref{eq:23-2}. The Jacobian of \eqref{eq:29}--\eqref{eq:29-2} is given by
\begin{align}
  \label{eq:30}
  J(\xi)=
  \begin{pmatrix}
    -e^{\xi_1}\beta'(e^{\xi_1}-1)&1\\
2e^{\xi_1}\gamma(e^{\xi_1}-1)\gamma'(e^{\xi_1}-1)&-2-2\xi_2
  \end{pmatrix}
\end{align}
where $\beta'(\sigma)=\frac{1}{(1+\sigma)^2}$. Let $v=(2\bar{\gamma}+\epsilon,1)$ where $\epsilon$ is chosen small enough so that $c_1:=(2\bar{\gamma}+\epsilon-2)<0$. It follows that
\begin{align}
  \label{eq:14}
  v^TJ(\xi)\leq (-\epsilon e^{-\xi_1},c_1)< 0\quad \forall \xi.
\end{align}
Applying Corollary \ref{cor:1}, 
the origin of \eqref{eq:23}--\eqref{eq:23-2} and \eqref{eq:29}--\eqref{eq:29-2} is globally asymptotically stable. Furthermore, we have the following state and flow sum-separable Lyapunov functions for the comparison system \eqref{eq:29}--\eqref{eq:29-2}:
\begin{align}
  \label{eq:50}
 V_1(\xi)&=(2\bar{\gamma}+\epsilon)\xi_1+\xi_2{\text{, and}}  \\
V_2(\xi)&=(2\bar{\gamma}+\epsilon)|\dot{\xi}_1|+|\dot{\xi_2}|.
\end{align}
Above, we understand $\dot{\xi}_1$ and $\dot{\xi}_2$ to be shorthand for the equalities expressed in \eqref{eq:29}--\eqref{eq:29-2}.
\end{example}

\begin{example}[Multiagent system]
\label{ex:multiagent}
Consider the following system evolving on $\Domain=\mathbb{R}^3$:
\begin{align}
  \label{eq:31}
  \dot{x}_1&=-\alpha_1(x_1)+\rho_1(x_3-x_1)\\
\dot{x}_2&=\rho_2(x_1-x_2)+\rho_3(x_3-x_2)\\
  \label{eq:31-3}\dot{x}_3&=\rho_4(x_2-x_3)
\end{align}
where we assume $\alpha_1:\mathbb{R}\to\mathbb{R}$ is strictly increasing and satisfies $\alpha(0)=0$  and $\alpha_1'(\sigma)\geq \ul{c}_0$ for some $\ul{c}_0>0$ for all $\sigma$, and each $\rho_i:\mathbb{R}\to\mathbb{R}$ is strictly increasing and satisfies $\rho_i(0)=0$. Furthermore, for $i=1,3$, $\rho'_i(\sigma)\leq \ol{c}_i$ for some $\ol{c}_i>0$ for all $\sigma$, and for $i=2,4$, $\rho'_i(\sigma)\geq \ul{c}_i$ for some $\ol{c}_i>0$ for all $\sigma$. 

For example, $x_1$, $x_2$, and $x_3$ may be the position of three vehicles, for which the dynamics  \eqref{eq:31}--\eqref{eq:31-3}  are a rendezvous protocol whereby agent 1 moves towards agent 3 at a rate dependent on the distance $x_3-x_1$ as determined by $\rho_1$, \emph{etc.} Additionally, agent 1 navigates towards the origin according to $-\alpha_1(x_1)$. Computing the Jacobian, we obtain
\begin{align}
  \label{eq:32}
\nonumber&  J(x)=\\
  &\begin{pmatrix}
    -\alpha'(x_1)-\rho_1'(z_{31})&0&\rho_1'(z_{31})\\
\rho_2'(z_{12})&-\rho_2'(z_{12})-\rho_3'(z_{32})&\rho_3'(z_{32})\\
0&\rho'_4(z_{23})&-\rho'_4(z_{23})
  \end{pmatrix}
\end{align}
where $z_{ij}:= x_i-x_j$. Let $w=(1,1+\epsilon_1,1+\epsilon_1+\epsilon_2)^T$
where $\epsilon_1>0$ and $\epsilon_2>0$ are chosen to satisfy
\begin{align}
  \label{eq:33}
\ul{c}_0&>  (\epsilon_1+\epsilon_2)\ol{c}_1\quad \text{and}\quad \epsilon_1\ul{c}_2>\epsilon_2\ol{c}_3.%
\end{align}
We then have $  J(x)w\leq c\mathbf{1}$ for all $x$ 
for $c=\max\{(\epsilon_1+\epsilon_2) \ol{c}_1-\ul{c}_0 ,\epsilon_2\ol{c}_3-\epsilon_1\ul{c}_2,-\epsilon_2\ul{c}_4\}<0$. Thus, the origin of \eqref{eq:31}--\eqref{eq:31-3} is globally asymptotically stable by Corollary \ref{cor:2}. Furthermore,
\begin{align}
  \label{eq:37}
  V_1(x)&=\max\{|x_1|,(1+\epsilon_1)^{-1}|x_2|,(1+\epsilon_1+\epsilon_2) ^{-1}|x_3|\},\\
  \label{eq:37-2}   V_2(x)&=\max\{|\dot{x}_1|,(1+\epsilon_1) ^{-1}|\dot{x}_2|,(1+\epsilon_1+\epsilon_2) ^{-1}|\dot{x}_3|\}
\end{align}
are state and flow max-seperable Lyapunov functions where we interpret $\dot{x}_i$ as shorthand for the equalities expressed in \eqref{eq:31}--\eqref{eq:31-3}. %
Since we may take $\epsilon_1$ and $\epsilon_2$ arbitrarily small satisfying \eqref{eq:33},  using Proposition \ref{prop:limit} we have also the following choices for Lyapunov functions:
\begin{align}
  \label{eq:49}
  V_3(x)&=\max\{|x_1|,|x_2|,|x_3|\},\\
  \label{eq:49-2}   V_4(x)&=\max\{|\dot{x}_1|,|\dot{x}_2|,|\dot{x}_3|\} .
\end{align}

The flow max-separable Lyapunov functions \eqref{eq:37-2} and \eqref{eq:49-2} are particularly useful for multiagent vehicular networks where it often easier to measure each agent's velocity rather than absolute position.
\end{example}
In Example \ref{ex:multiagent}, choosing $w=\mathbf{1}$, we have $J(x)w\leq 0$ for all $x$, however this is not enough to establish asymptotic stability using Corollary \ref{cor:2}. Informally, choosing $w$ as in the example
distributes the extra negativity of $-\alpha'(x_1)$ among the columns of $J(x)$. Nonetheless, Proposition \ref{prop:limit} implies that choosing $w=\mathbf{1}$ indeed leads to a valid Lyapunov function.%

The above example generalizes to systems with many agents interacting via arbitrary directed graphs, as does the principle of distributing extra negativity  along diagonal entries of the Jacobian as discussed in Section \ref{sec:disc-comp-exist}.

\begin{example}[Traffic flow]
\label{ex:traffic}
A model of traffic flow along a freeway with no onramps is obtained by spatially partitioning the freeway into $n$ segments such that traffic flows from segment $i$ to $i+1$, $x_i\in[0,\bar{x}_i]$ is the density of vehicles occupying link $i$, and $\bar{x}_i$ is the capacity of link $i$. A fraction $\beta_i\in(0,1]$ of the flow out of link $i$ enters link $i+1$. The remaining $1-\beta_i$ fraction is assumed to exit the network via, \emph{e.g.}, unmodeled offramps. Associated with each link is a continuously differentiable \emph{demand} function $D_i:[0,\bar{x}_i]\to\mathbb{R}_{\geq 0}$ that is strictly increasing and satisfies $D_i(0)=0$, and a continuously differentiable \emph{supply} function $S_i:[0,\bar{x}_i]\to\mathbb{R}_{\geq 0}$ that is strictly decreasing and satisfies $S_i(\bar{x}_i)=0$. Flow from segment to segment is restricted by upstream demand and downstream supply, and the change in density of a link is governed by mass conservation:
\begin{align}
  \label{eq:38}
\dot{x}_1&= \min\{\delta_1,S_1(x_1)\}-\frac{1}{\beta_1}g_{1}(x_{1},x_{2})\\
  \dot{x}_i&= g_{i-1}(x_{i-1},x_i)-\frac{1}{\beta_i}g_{i}(x_{i},x_{i+1}), \quad i=2,\ldots,n-1\\
  \label{eq:38-3}\dot{x}_n&=g_{n-1}(x_{n-1},x_n)- D_n(x_n)
\end{align}
for some $\delta_1>0$ where, for $i=1,\ldots,n-1$,
\begin{align}
  \label{eq:39}
g_{i}(x_{i},x_{i+1})=\min\{\beta_iD_i(x_i),S_{i+1}(x_{i+1})\}.
\end{align}
Let $\delta_i:= \delta_1\prod_{j=1}^{i-1}\beta_j$ for $i=2,\ldots, n$. If $d^{-1}_i(\delta_i)<s^{-1}_i(\delta_i)$ for all $i$, then $\delta_1$ is said to be \emph{feasible} and $x^*_i:=d^{-1}_i(\delta_i)$ constitutes the unique equilibrium. 

\begin{figure*}
\begin{align}
  \label{eq:36}
  J(x)=
  \begin{pmatrix}
    \partial_1g_0-\frac{1}{\beta_1}\partial_1g_1&-\frac{1}{\beta_1}\partial_2g_1&0&0&\cdots&0\\
\partial_1 g_1&\partial_2g_1-\frac{1}{\beta_2}\partial_2g_2&-\frac{1}{\beta_2}\partial_3 g_2&0 &\cdots&0\\
0&\partial_2g_2&\partial_3g_2-\frac{1}{\beta_3}\partial_3g_3&-\frac{1}{\beta_3}\partial_4g_3&&0\\
\vdots&&&&\ddots&\vdots\\
0&0&\cdots&0&\partial_{n-1}g_{n-1}&\partial_{n}g_{n-1}-\partial_nD_n(x_n)
  \end{pmatrix}
\end{align}
\hrule
\end{figure*}

Let $\partial_i$ denote differentiation with respect to the $i$th component of $x$, that is, $\partial_ig(x):=\frac{\partial g}{\partial x_i}(x)$ for a function $g(x)$. The dynamics \eqref{eq:38}--\eqref{eq:38-3} define a system $\dot{x}=f(x)$ for which $f$ is continuous but only piecewise differentiable. Nonetheless, the results developed above apply for this case, and, in the sequel, we interpret statements involving derivatives to hold wherever the derivative exists.

Notice that $\partial_{i}g_i(x_i,x_{i+1})\geq 0$ and $\partial_{i+1}g_i(x_i,x_{i+1})\leq 0$. Define $g_0(x_1):=\min\{\delta_1,S_1(x_1)\} $. 
The Jacobian, where it exists, is given by \eqref{eq:36}, which is seen to be Metzler. Let
\begin{align}
  \label{eq:40}
  \tilde{v}=\begin{pmatrix}1,\beta_1^{-1},(\beta_1\beta_2)^{-1},\ldots,(\beta_1\beta_2\cdots\beta_{n-1})^{-1}\end{pmatrix}^T.
\end{align}
Then $\tilde{v}^TJ(x)\leq 0$ for all $x$. Moreover, there exists $\epsilon=(\epsilon_1,\epsilon_2,\ldots,\epsilon_{n-1},0)$ with $\epsilon_{i}>\epsilon_{i+1}$ for each $i$ such that $v:=\tilde{v}+\epsilon$ satisfies
\begin{align}
  \label{eq:41}
  v^TJ(x)&\leq 0\quad \forall x\\
  \label{eq:41-2} v^TJ(x^*)&<0.
\end{align}
Such a vector $\epsilon$ is constructed using a technique similar to that used in Example \ref{ex:multiagent}. In particular, the sum of the $n$th column of $\text{diag}(\tilde{v})J(x)$ is strictly negative because $-\partial_nD_n(x_n)<0$, and this excess negativity is used to construct $v$ such that \eqref{eq:41}--\eqref{eq:41-2} holds. A particular choice of $\epsilon$ such that \eqref{eq:41}--\eqref{eq:41-2} holds depends on bounds on the derivative of the demand functions $D_i$, but it is possible to choose $\epsilon$ arbitrarily small. Corollary~\ref{cor:1} establishes asymptotic stability, and Proposition \ref{prop:limit} gives the following sum-separable Lyapunov functions:
\begin{align}
  \label{eq:42}
  V_1(x)&=\sum_{i=1}^n \left(x_i\prod_{j=1}^{i-1}\beta_j\right),\\
  \label{eq:42-2}  V_2(x)&=\sum_{i=1}^n \left(|\dot{x}_i|\prod_{j=1}^{i-1}\beta_j\right),
\end{align}
where we interpret $\dot{x}_i$ according to \eqref{eq:38}--\eqref{eq:38-3}.

In traffic networks, it is often easier to measure traffic flow rather than traffic density. Thus, \eqref{eq:42-2} is a practical Lyapunov function indicating that the (weighted) total absolute net flow throughout the network decreases over time.
\end{example}
 In \cite{coogan2015compartmental}, a result similar to that of Example \ref{ex:traffic} is derived for possibly infeasible input flow and  traffic flow network topologies where merging junctions with multiple incoming links are allowed. The proof considers a flow sum-separable Lyapunov function similar to \eqref{eq:42-2} and appeals to LaSalle's invariance principle rather than Proposition \ref{prop:limit}.

\section{Discussion}
\label{sec:disc-comp-exist}
{For nonlinear monotone systems with an asymptotically stable equilibrium, it is shown in \cite{Dirr:2015rt} that max-separable Lyapunov functions can always be constructed in compact invariant subsets of the domain of attraction. Such a Lyapunov function is constructed by considering a single dominating trajectory that is (componentwise) greater than all points in this subset. If there exists a trajectory of the system that converges to the equilibrium in forward time and diverges to infinity in all components in backwards time, then this construction leads to a global Lyapunov function. It is also shown in \cite{Sootla:2016sp} that such max-separable Lyapunov functions can be obtained from the leading eigenfunction of the linear, but infinite dimensional, Koopman operator associated with the monotone system. 

In contrast, Theorem \ref{thm:tvmain2} and Corollary \ref{cor:2} above only provide sufficient conditions for constructing max-separable Lyapunov functions. However, these results offer an alternative construction than that suggested in \cite{Dirr:2015rt, Sootla:2016sp}.

In addition, \cite{Dirr:2015rt} provides counterexamples showing that an asymptotically stable monotone system need not admit sum-separable or max-separable Lyapunov functions globally \cite{Dirr:2015rt}.  An important remaining open question is whether monotone systems whose domain satisfies a compact invariance condition necessarily admit sum-separable Lyapunov functions. 


}
\subsection{Relationship to ISS small-gain conditions}
\label{sec:relat-iss-small}
In this section, we briefly discuss the relationship between the main results of this paper and small-gain conditions for interconnected input-to-state stable (ISS) systems. While there appears to be a number of interesting connections between the results presented here and the extensive literature on networks of ISS systems, see, \emph{e.g.},  \cite{Ito:2012ux,Ito:2013ez} for some recent results, a complete characterization of this relationship is outside the scope of this paper and will be the subject of future research.  Nonetheless, we highlight how Theorems \ref{thm:tvmain1} and \ref{thm:tvmain2} provide a Jacobian-based perspective to ISS system analysis.

Consider $N$ interconnected systems with dynamics $\dot{x}_i=f_i(x_1,\ldots,x_N)$ for $x_i\in\mathbb{R}^{n_i}$ and suppose each system satisfies an input-to-state stability (ISS) condition \cite{Sontag:1989fk} whereby there exists ISS Lyapunov functions $V_i$ \cite{Sontag:1995qf} satisfying
\begin{align}
  \label{eq:57}
  \frac{\partial V_i}{\partial x_i}(x_i)f_i(x)\leq -\alpha_i(V_i(x_i))+\sum_{i\neq j}\gamma_{ij}(V_j(x_j))
\end{align}
where each $\alpha_i$ and $\gamma_{ij}$ is a class $\Kinf$ function\footnote{A continuous function $\alpha:\mathbb{R}_{\geq 0}\to\mathbb{R}_{\geq 0}$ is of class $\Kinf$ if it is strictly increasing, $\alpha(0)=0$, and $\lim_{r\to\infty}\alpha(r)=\infty$.}. We obtain a monotone comparison system 
\begin{align}
  \label{eq:51}
\dot{\xi}=g(\xi), \qquad g_i(\xi)=-\alpha_i(\xi_i)+\sum_{j\neq i}\gamma_{ij}(\xi_j)
\end{align}
evolving on $\mathbb{R}^n_{\geq 0}$ for which asymptotic stability of the origin implies asymptotic stability of the original system \cite{Ruffer:2010tw}.

It is shown in \cite[Section 4.3]{Dashkovskiy:2011qv} that if $\gamma_{ij}(s)=k_{ij}h_j(s)$ and $\alpha_i(s)=a_ih_i(s)$ for some collection of constants $c_{ij}\geq0 $, $a_j>0$ and $\Kinf$ functions $h_i$ for all $i,j=1,\ldots,n$, and there exists a vector $v$ such that $v^T(-A+C)<0$ where $A=\text{diag}(a_1,\ldots,a_n)$ and $[C]_{ij}=c_{ij}$, then $V(x)=v^T\begin{bmatrix}V_1(x_1)&\ldots V_n(x_n)\end{bmatrix}^T$ is an ISS Lyapunov function for the composite system. Indeed, in this case, and considering the comparison system \eqref{eq:51}, we see that
\begin{align}
  \label{eq:86}
  \frac{\partial g}{\partial \xi}(\xi)= (-A+C)\text{diag}(h'_1(\xi_1),\cdots,h'_n(\xi_n))
\end{align}
where $h'_i(\xi_i)\geq 0$ so that, if $v^T(-A+C)<0$, then also $v^T\frac{\partial g}{\partial \xi}(\xi)\leq 0$ for all $\xi$. If also $h'_i(0)>0$ for all $i$, then $v^T\frac{\partial g}{\partial \xi}(0)< 0$ so that Corollary \ref{cor:1} 
implies the sum-separable Lyapunov function $v^T\xi$, providing a contraction theoretic interpretation of this result. The case of $N=2$ was first investigated in \cite{Jiang:1996dw} where it is assumed without loss of generality that $a_1=a_2=1$ and it is shown that if $c_{12}c_{21}<1$, then $v_1V_1(x_1)+v_2V(x_2)$ is a Lyapunov function for the original system for any $v=\begin{bmatrix}v_1& v_2\end{bmatrix}^T>0$ satisfying $v_1c_{12}<v_2$ and $v_2c_{21}<v_1$. These conditions are equivalent to $v^T(-I+C)<0$.

 Alternatively, in \cite{Ruffer:2010tw, Dashkovskiy:2010zh}, it is shown that if there exists a function $\rho:\mathbb{R}_{\geq 0}\to\mathbb{R}_{\geq 0}^n$ with each component $\rho_i$ belonging to class $\mathcal{K}_\infty$ such that $g(\rho(r))< 0$ for all $r>0$, then the origin is asymptotically stable and $V(\xi):=\max_i\{\rho_i^{-1}(\xi_i)\}$ is a Lyapunov function. If the conditions of Corollary \ref{cor:2} hold for the comparison system for some $w$, we may choose $\rho(r)=rw$. Indeed, we have
\begin{align}
  \label{eq:55}
  g(rw)=\int_0^1\frac{\partial g}{\partial \xi}(\sigma rw) rw\ d \sigma <0 \quad \forall r>0.
\end{align}
 For this case, $V(\xi)=\max_i\{\rho_i^{-1}(\xi_i)\}=\max_i\{\xi_i/w_i\}$, recovering \eqref{eq:19}.

\subsection{Generalized contraction and compartmental systems}
We now discuss the relationship between the results presented here and additional results for contractive systems in the literature. First, we comment on the relationship between Corollary \ref{cor:nonexpeq} of Theorem \ref{thm:2} as well as Proposition \ref{prop:limit} and a generalization of contraction theory recently developed in \cite{Sontag:2014eu, Margaliot:2015wd} where exponential contraction between any two trajectories is required only after an arbitrarily small amount of time, an arbitrarily small overshoot, or both. In \cite[Corollary 1]{Margaliot:2015wd}, it is shown that if a system is contractive with respect to a sequence of norms convergent to some norm, then the system is generalized contracting with respect to that norm, a result analogous to Proposition \ref{prop:limit}. In \cite{Margaliot:2015wd}, conditions on the sign structure of the Jacobian are obtained that ensure the existence of such a sequence of weighted $\ell_1$ or $\ell_\infty$ norms. These conditions are a generalization of the technique in Example \ref{ex:multiagent}
and Example \ref{ex:traffic} 
in Section \ref{sec:examples} where small $\epsilon$ is used to distribute excess negativity.  

Furthermore, it is shown in \cite{Margaliot:2012hc, Margaliot:2014qv} that a ribosome flow model for gene translation is monotone and nonexpansive with respect to a weighted $\ell_1$ norm, and additionally is contracting on a subset of its domain. %
Entrainment of solutions is proved by first showing that all trajectories reach the region of exponential contraction. Theorem \ref{thm:2} provides a different approach for studying entrainment by observing that the distance to the periodic trajectory strictly decreases in each period due to a neighborhood of contraction along the periodic trajectory.%

Finally, we note that Metzler matrices with nonpositive column sums have also been called \emph{compartmental} \cite{Jacquez:1993uq}. It has been shown that if the Jacobian matrix is compartmental for all $x$, then $V(x)=|f(x)|$ is a nonincreasing function along trajectories of \eqref{eq:1} \cite{Jacquez:1993uq, Maeda:1978fk}; Proposition \ref{thm:1} recovers this observation when considering \eqref{eq:71} with $c=0$, $\Theta(x)\equiv I$, and $|\wc|$ taken to be the $\ell_1$ norm.

\section{Conclusions}
\label{sec:conclusions}
We have investigated monotone systems that are also contracting with respect to a weighted $\ell_1$ norm or $\ell_\infty$ norm. In the case of the $\ell_1$ (respectively, $\ell_\infty$) norm, we provided a condition on the weighted column (respectively, row) sums of the Jacobian matrix for ensuring contraction. When the norm is state-dependent, these conditions include an additive term that is the time derivative of the weights. This construction leads to a pair of sum-separable (respectively, max-separable) Lyapunov functions. The first Lyapunov function pair is separable along the state of the system while the second is \emph{agent-separable}, that is, each constituent function depends on $f_i(x)$ in addition to $x_i$. When the weighted contractive norm is independent of state, the components of this Lyapunov function only depend on $f_i(x)$ and we say it is \emph{flow-separable}. Such flow separable Lyapunov functions are especially relevant in applications where it is easier to measure the derivative of the system's state rather than measure the state directly.

In addition, we provided a computational algorithm to search for separable Lyapunov functions using our main results and sum-of-squares programming, and we demonstrated our results through several examples. We further highlighted some connections to stability results for interconnected input-to-state stable systems. These connections appear to be a promising direction for future work. %


%

%

%

%
%

%
%

%
\bibliography{$HOME/Documents/Books/books}

\begin{thebibliography}{}

\bibitem[Angeli and Sontag, 2003]{Angeli:2003fv}
Angeli, D. and Sontag, E. (2003).
\newblock Monotone control systems.
\newblock {\em IEEE Transactions on Automatic Control}, 48(10):1684--1698.

\bibitem[Angeli and Sontag, 2004]{Angeli:2004qy}
Angeli, D. and Sontag, E. (2004).
\newblock Interconnections of monotone systems with steady-state
  characteristics.
\newblock In {\em Optimal control, stabilization and nonsmooth analysis}, pages
  135--154. Springer.

\bibitem[Bao et~al., 2012]{Bao:2012mg}
Bao, D., Chern, S.-S., and Shen, Z. (2012).
\newblock {\em An introduction to Riemann-Finsler geometry}, volume 200.
\newblock Springer Science \& Business Media.

\bibitem[Boyd et~al., 1994]{Boyd:1994uq}
Boyd, S., Ghaoui, L.~E., Feron, E., and Balakrishnan, V. (1994).
\newblock {\em Linear Matrix Inequalities in System and Control Theory}.
\newblock SIAM.

\bibitem[Clarke, 1990]{Clarke:1990kx}
Clarke, F.~H. (1990).
\newblock {\em Optimization and nonsmooth analysis}, volume~5.
\newblock Siam.

\bibitem[Como et~al., 2015]{Como:2015ne}
Como, G., Lovisari, E., and Savla, K. (2015).
\newblock Throughput optimality and overload behavior of dynamical flow
  networks under monotone distributed routing.
\newblock {\em IEEE Transactions on Control of Network Systems}, 2(1):57--67.

\bibitem[Coogan, 2016]{Coogan:2016kx}
Coogan, S. (2016).
\newblock Separability of {L}yapunov functions for contractive monotone
  systems.
\newblock In {\em IEEE Conference on Decision and Control}, pages 2184--2189.

\bibitem[Coogan and Arcak, 2014]{Coogan:2014ph}
Coogan, S. and Arcak, M. (2014).
\newblock Dynamical properties of a compartmental model for traffic networks.
\newblock In {\em Proceedings of the 2014 American Control Conference}, pages
  2511--2516.

\bibitem[Coogan and Arcak, 2015a]{coogan2015compartmental}
Coogan, S. and Arcak, M. (2015a).
\newblock A compartmental model for traffic networks and its dynamical
  behavior.
\newblock {\em IEEE Transactions on Automatic Control}, 60(10):2698--2703.

\bibitem[Coogan and Arcak, 2015b]{Coogan:2015dq}
Coogan, S. and Arcak, M. (2015b).
\newblock A dissipativity approach to safety verification for interconnected
  systems.
\newblock {\em IEEE Transactions on Automatic Control}, 60(6):1722--1727.

\bibitem[Coogan and Arcak, 2015c]{Coogan:2014ty}
Coogan, S. and Arcak, M. (2015c).
\newblock Efficient finite abstraction of mixed monotone systems.
\newblock In {\em Proceedings of the 18th International Conference on Hybrid
  Systems: Computation and Control}, pages 58--67.

\bibitem[Dashkovskiy et~al., 2011]{Dashkovskiy:2011qv}
Dashkovskiy, S., Ito, H., and Wirth, F. (2011).
\newblock On a small gain theorem for {ISS} networks in dissipative lyapunov
  form.
\newblock {\em European Journal of Control}, 17(4):357--365.

\bibitem[Dashkovskiy et~al., 2010]{Dashkovskiy:2010zh}
Dashkovskiy, S.~N., R\"{u}ffer, B.~S., and Wirth, F.~R. (2010).
\newblock Small gain theorems for large scale systems and construction of {ISS}
  {L}yapunov functions.
\newblock {\em SIAM Journal on Control and Optimization}, 48(6):4089--4118.

\bibitem[Desoer and Vidyasagar, 2008]{Desoer:2008bh}
Desoer, C. and Vidyasagar, M. (2008).
\newblock {\em Feedback systems: Input-output properties}.
\newblock Society for Industrial and Applied Mathematics.

\bibitem[Dirr et~al., 2015]{Dirr:2015rt}
Dirr, G., Ito, H., Rantzer, A., and R{\"u}ffer, B. (2015).
\newblock Separable {L}yapunov functions for monotone systems: constructions
  and limitations.
\newblock {\em Discrete Contin. Dyn. Syst. Ser. B}.

\bibitem[Forni and Sepulchre, 2014]{Forni:2012qe}
Forni, F. and Sepulchre, R. (2014).
\newblock A differential {L}yapunov framework for contraction analysis.
\newblock {\em IEEE Transactions on Automatic Control}, 59(3):614--628.

\bibitem[Gomes et~al., 2008]{Gomes:2008fk}
Gomes, G., Horowitz, R., Kurzhanskiy, A.~A., Varaiya, P., and Kwon, J. (2008).
\newblock Behavior of the cell transmission model and effectiveness of ramp
  metering.
\newblock {\em Transportation Research Part C: Emerging Technologies},
  16(4):485--513.

\bibitem[Hirsch, 1983]{Hirsch:1983lq}
Hirsch, M.~W. (1983).
\newblock Differential equations and convergence almost everywhere in strongly
  monotone semiflows.
\newblock {\em Contemporary Mathematics}, 17:267--285.

\bibitem[Hirsch, 1985]{Hirsch:1985fk}
Hirsch, M.~W. (1985).
\newblock Systems of differential equations that are competitive or cooperative
  {II}: Convergence almost everywhere.
\newblock {\em SIAM Journal on Mathematical Analysis}, 16(3):423--439.

\bibitem[Ito et~al., 2012]{Ito:2012ux}
Ito, H., Dashkovskiy, S., and Wirth, F. (2012).
\newblock Capability and limitation of max-and sum-type construction of
  {L}yapunov functions for networks of {iISS} systems.
\newblock {\em Automatica}, 48(6):1197--1204.

\bibitem[Ito et~al., 2013]{Ito:2013ez}
Ito, H., Jiang, Z.~P., Dashkovskiy, S.~N., and R{\"u}ffer, B.~S. (2013).
\newblock Robust stability of networks of {iISS} systems: Construction of
  sum-type lyapunov functions.
\newblock {\em IEEE Transactions on Automatic Control}, 58(5):1192--1207.

\bibitem[Jacquez and Simon, 1993]{Jacquez:1993uq}
Jacquez, J.~A. and Simon, C.~P. (1993).
\newblock Qualitative theory of compartmental systems.
\newblock {\em SIAM Review}, 35(1):43--79.

\bibitem[Jiang et~al., 1996]{Jiang:1996dw}
Jiang, Z.-P., Mareels, I.~M., and Wang, Y. (1996).
\newblock A {L}yapunov formulation of the nonlinear small-gain theorem for
  interconnected {ISS} systems.
\newblock {\em Automatica}, 32(8):1211 -- 1215.

\bibitem[Lohmiller and Slotine, 1998]{LOHMILLER:1998bf}
Lohmiller, W. and Slotine, J.-J.~E. (1998).
\newblock On contraction analysis for non-linear systems.
\newblock {\em Automatica}, 34(6):683--696.

\bibitem[Lovisari et~al., 2014a]{Lovisari:2014qv}
Lovisari, E., Como, G., Rantzer, A., and Savla, K. (2014a).
\newblock Stability analysis and control synthesis for dynamical transportation
  networks.
\newblock {\em arXiv:1410.5956}.

\bibitem[Lovisari et~al., 2014b]{Lovisari:2014yq}
Lovisari, E., Como, G., and Savla, K. (2014b).
\newblock Stability of monotone dynamical flow networks.
\newblock In {\em Proceedings of the 53rd Conference on Decision and Control},
  pages 2384--2389.

\bibitem[Maeda et~al., 1978]{Maeda:1978fk}
Maeda, H., Kodama, S., and Ohta, Y. (1978).
\newblock Asymptotic behavior of nonlinear compartmental systems:
  nonoscillation and stability.
\newblock {\em IEEE Transactions on Circuits and Systems}, 25(6):372--378.

\bibitem[Manchester and Slotine, 2017]{Manchester:2017wc}
Manchester, I.~R. and Slotine, J.-J.~E. (2017).
\newblock On existence of separable contraction metrics for monotone nonlinear
  systems.
\newblock {\em arxiv:1704.02676}.

\bibitem[Margaliot et~al., 2014]{Margaliot:2014qv}
Margaliot, M., Sontag, E.~D., and Tuller, T. (2014).
\newblock Entrainment to periodic initiation and transition rates in a
  computational model for gene translation.
\newblock {\em PloS one}, 9(5):e96039.

\bibitem[Margaliot et~al., 2016]{Margaliot:2015wd}
Margaliot, M., Sontag, E.~D., and Tuller, T. (2016).
\newblock Contraction after small transients.
\newblock {\em Automatica}, 67:178--184.

\bibitem[Margaliot and Tuller, 2012]{Margaliot:2012hc}
Margaliot, M. and Tuller, T. (2012).
\newblock Stability analysis of the ribosome flow model.
\newblock {\em IEEE/ACM Transactions on Computational Biology and
  Bioinformatics (TCBB)}, 9(5):1545--1552.

\bibitem[M.Vidyasagar, 2002]{Vidyasagar:2002ly}
M.Vidyasagar (2002).
\newblock {\em Nonlinear System Analysis}.
\newblock Society for Industrial and Applied Mathematics.

\bibitem[Papachristodoulou and Prajna, 2002]{Papachristodoulou:2002jk}
Papachristodoulou, A. and Prajna, S. (2002).
\newblock On the construction of {L}yapunov functions using the sum of squares
  decomposition.
\newblock In {\em Proceedings of the 41st IEEE Conference on Decision and
  Control}, volume~3, pages 3482--3487 vol.3.

\bibitem[Parrilo, 2000]{Parrilo:2000fk}
Parrilo, P. (2000).
\newblock {\em Structured Semidefinite Programs and Semialgebraic Geometry
  Methods in Robustness and Optimization}.
\newblock PhD thesis, California Institute of Technology.

\bibitem[Parrilo, 2003]{Parrilo:2001uq}
Parrilo, P.~A. (2003).
\newblock Semidefinite programming relaxations for semialgebraic problems.
\newblock {\em Mathematical Programming Ser. B}, 96(2):293--320.

\bibitem[Pavlov et~al., 2004]{Pavlov:2004lr}
Pavlov, A., Pogromsky, A., van~de Wouw, N., and Nijmeijer, H. (2004).
\newblock Convergent dynamics, a tribute to {B}oris {P}avlovich {D}emidovich.
\newblock {\em Systems \& Control Letters}, 52(3):257--261.

\bibitem[Prajna et~al., 2004]{sostools}
Prajna, S., Papachristodoulou, A., Seiler, P., and Parrilo, P.~A. (2004).
\newblock {\em {SOSTOOLS}: Sum of squares optimization toolbox for {MATLAB}}.

\bibitem[Rantzer, 2015]{Rantzer:2012fj}
Rantzer, A. (2015).
\newblock Scalable control of positive systems.
\newblock {\em European Journal of Control}, 24:72--80.

\bibitem[Rantzer et~al., 2013]{Rantzer:2013bf}
Rantzer, A., Ruffer, B., and Dirr, G. (2013).
\newblock Separable {L}yapunov functions for monotone systems.
\newblock In {\em Decision and Control (CDC), 2013 IEEE 52nd Annual Conference
  on}, pages 4590--4594.

\bibitem[Raveh et~al., 2016]{Raveh:2015wm}
Raveh, A., Margaliot, M., Sontag, E.~D., and Tuller, T. (2016).
\newblock A model for competition for ribosomes in the cell.
\newblock {\em Journal of The Royal Society Interface}, 13(116):20151062.

\bibitem[R\"{u}ffer et~al., 2010]{Ruffer:2010tw}
R\"{u}ffer, B.~S., Kellett, C.~M., and Weller, S.~R. (2010).
\newblock Connection between cooperative positive systems and integral
  input-to-state stability of large-scale systems.
\newblock {\em Automatica}, 46(6):1019--1027.

\bibitem[Smith, 1995]{Smith:2008fk}
Smith, H.~L. (1995).
\newblock {\em Monotone dynamical systems: {A}n introduction to the theory of
  competitive and cooperative systems}.
\newblock American Mathematical Society.

\bibitem[Sontag, 1989]{Sontag:1989fk}
Sontag, E. (1989).
\newblock Smooth stabilization implies coprime factorization.
\newblock {\em IEEE Transactions on Automatic Control}, 34(4):435--443.

\bibitem[Sontag, 1998]{Sontag:1998cr}
Sontag, E.~D. (1998).
\newblock {\em Mathematical Control Theory: Deterministic Finite Dimensional
  Systems}.
\newblock Springer, second edition.

\bibitem[Sontag, 2007]{Sontag:2007ad}
Sontag, E.~D. (2007).
\newblock Monotone and near-monotone biochemical networks.
\newblock {\em Systems and Synthetic Biology}, 1(2):59--87.

\bibitem[Sontag, 2010]{Sontag:2010fk}
Sontag, E.~D. (2010).
\newblock Contractive systems with inputs.
\newblock In {\em Perspectives in Mathematical System Theory, Control, and
  Signal Processing}, pages 217--228. Springer.

\bibitem[Sontag et~al., 2014]{Sontag:2014eu}
Sontag, E.~D., Margaliot, M., and Tuller, T. (2014).
\newblock On three generalizations of contraction.
\newblock In {\em IEEE 53rd Annual Conference on Decision and Control (CDC)},
  pages 1539--1544.

\bibitem[Sontag and Wang, 1995]{Sontag:1995qf}
Sontag, E.~D. and Wang, Y. (1995).
\newblock On characterizations of the input-to-state stability property.
\newblock {\em Systems \& Control Letters}, 24(5):351--359.

\bibitem[Sootla, 2016]{Sootla:2016sp}
Sootla, A. (2016).
\newblock Construction of max-separable {L}yapunov functions for monotone
  systems using the {K}oopman operator.
\newblock In {\em 2016 IEEE 55th Conference on Decision and Control (CDC)},
  pages 6512--6517.

\bibitem[Topcu et~al., 2007]{Topcu:2007fk}
Topcu, U., Packard, A., Seiler, P., and Wheeler, T. (2007).
\newblock Stability region analysis using simulations and sum-of-squares
  programming.
\newblock In {\em Proceedings of the American Control Conference}, pages
  6009--6014.

\end{thebibliography}
\end{document}